\pdfoutput=1 % force arXiv to use pdflatex

\documentclass[
  a4paper,
  UKenglish,
  cleveref,
  autoref,
  thm-restate
]{lipics-v2021}

\hideLIPIcs % removes LIPIcs banner/DOI for arXiv preprint

\newif\ifOL
\OLfalse

\usepackage[T1]{fontenc}
\usepackage[utf8]{inputenc}

\usepackage{booktabs}
\usepackage{complexity}
\usepackage{tikz}
\usetikzlibrary{shapes,arrows.meta,positioning,calc}

\usepackage{xspace}
\usepackage[most]{tcolorbox}
\usepackage{tabularx}
\usepackage{pifont}
\usepackage{boxedminipage}

\definecolor{Carmine}{RGB}{150,0,24}
\definecolor{LightBrownDrab}{RGB}{181,165,135}
\definecolor{BrownDrab}{RGB}{150,113,23}

\newcommand{\cmark}{\ding{51}}

\newcommand{\mcc}{\textsc{Multicolored Clique}\xspace}

\newcommand{\msr}{\textsc{Min-Sum-Radii}}

\newcommand{\exact}{\textsc{Exact-}\textsc{MSR}}
\newcommand{\fixed}{\textsc{Allowed-Centers-}\textsc{MSR}}

\newcommand{\yes}{\normalfont\textsc{Yes}\xspace}

\newcommand{\highl}{+}
\newcommand{\lowl}{-}

\newcommand{\cost}{\Delta}
\newcommand{\CCC}{\mathcal{C}}
\newcommand{\bigoh}{\mathcal{O}}

\newcommand{\RR}{\mathbb{R}}

\newcommand{\SM}{{\;{|}\;}}
\renewcommand{\SE}{\,\}}
\newcommand{\SB}{\{\,}

\newcommand{\td}{\textsf{td}}
\newcommand{\tw}{\textsf{tw}}
\newcommand{\vc}{\textsf{vc}}
\newcommand{\fvs}{\textsf{fvs}}

\newcommand{\MSR}{MSR}

\newcommand{\high}{\textsc{$x^i_{\highl}$}}
\newcommand{\low}{\textsc{$x^i_{\lowl}$}}

\newcommand{\ETH}{\ensuremath{\mathsf{ETH}}\xspace}

\newcommand{\pbDef}[3]{%
  \noindent
  \begin{center}
  \begin{boxedminipage}{0.98\columnwidth}
  {\sc #1}\\[5pt]
  \begin{tabular}{l p{0.70\columnwidth}}
  {\sc Instance}: & #2\\
  {\sc Question}: & #3
  \end{tabular}
  \end{boxedminipage}
  \end{center}
}

\newtcolorbox{intuitionbox}{
  colback=gray!10,
  boxrule=0pt,
  frame hidden,
  left=6pt,
  right=6pt,
  top=6pt,
  bottom=6pt,
  arc=2pt
}

\newcommand{\pbDefP}[4]{%
  \noindent
  \begin{center}
  \begin{boxedminipage}{\columnwidth}
  {\sc #1}\\[5pt]
  \begin{tabular}{lp{0.8\columnwidth}}
  {\sc Instance}: & #2\\
  {\sc Parameter}: & #3\\
  {\sc Question}: & #4
  \end{tabular}
  \end{boxedminipage}
  \end{center}
}

\title{On the Parameterized Complexity of \msr}

\author{Pankaj Kumar}
{University of Birmingham, UK}
{pundir.pankaj25@gmail.com}
{}
{}

\author{Haiko M\"uller}
{University of Leeds, UK}
{H.Muller@leeds.ac.uk}
{}
{}

\author{Sebastian Ordyniak}
{University of Leeds, UK}
{sordyniak@gmail.com}
{}
{}

\author{Melanie Schmidt}
{Heinrich Heine University D\"usseldorf, Germany}
{mschmidt@hhu.de}
{}
{Melanie Schmidt acknowledges funding by DFG grant 456558332.}

\authorrunning{P. Kumar, H. M\"uller, S. Ordyniak, and M. Schmidt}

\Copyright{
Pankaj Kumar,
Haiko M\"uller,
Sebastian Ordyniak,
and Melanie Schmidt
}

\ccsdesc{Theory of computation~Parameterized complexity and exact algorithms}

\keywords{
Parameterized complexity,
Min-Sum-Radii clustering
}
\relatedversion{
An extended abstract of this paper will appear in the proceedings of the
20th Scandinavian Symposium on Algorithm Theory (SWAT 2026).
}
\nolinenumbers

\begin{document}

\maketitle

\begin{abstract}
In the \msr\ (\MSR) clustering problem, we are given a finite set $X$ of $n$ points in a metric space. The objective is to find at most $k$ clusters centered at a subset of these points such that every point of $X$ is assigned to one of the clusters, minimizing the sum of the radii of the clusters. 
The problem is known to be \NP-hard even on metrics induced by weighted planar graphs and metrics with constant doubling dimension, 
as shown by Gibson et al.~(SWAT 2008). In this work, we investigate the parameterized complexity of \MSR\ on metrics induced by undirected graphs. 
We distinguish between weighted graph metrics (with positive edge weights) 
and unweighted graph metrics (where all edges have unit weight). 

\textbf{Weighted Graph Metrics.} We show that \MSR\ is \W[1]-hard on
metrics induced by weighted bipartite graphs, when parameterized by
the combined parameter $k$ the number of clusters  and $\Delta$ the
cost of the clustering. 
We then investigate the structural
parameterized complexity of the problem. Drexler et
al.~[doi:10.48550/arXiv.2310.02130] showed that the \MSR\ problem admits an \XP\
algorithm on metrics induced by weighted graphs when parameterized by
treewidth, and asked whether this can be improved to fixed-parameter
tractability. We first answer their question in the negative, and more
strongly show that \MSR\ stays \textup{\W[1]\text{-hard}} on metrics
induced by undirected weighted bipartite graphs when parameterized by
the vertex cover number plus $k$. 
We then turn our attention to parameters for dense graphs and show
that \MSR\ remains \W[1]-hard when parameterized by $k+\Delta$
even on cliques and complete bipartite graphs. 

On the positive side, we employ the known \XP\ algorithm parameterized
by treewidth, to show that the \MSR\ problem is \FPT\ when
parameterized by the parameter treewidth plus $\Delta$. 
Together, these results provide a complete picture of the parameterized complexity of \MSR\ with respect to any combination of parameters $k$, $\cost$, as well as structural parameters for sparse graphs above vertex cover and known parameters for dense graphs (such as neighborhood diversity and modular width).

\textbf{Unweighted Graph Metrics.}
The story is rather different for unweighted graphs, since it is a long standing open question whether \MSR\ on metrics induced by undirected
graphs is solvable in polynomial-time. Although we cannot answer
this question, we provide classical and parameterized hardness results  for two very closely related problems, namely \exact\ (\MSR\ and one wants to find exactly $k$ clusters) and \fixed\ (\MSR\ with an additional set of allowed cluster centers). We also show that \MSR\ as well as these two problems are fixed-parameter tractable parameterized by the treedepth of the input graph.
\end{abstract}

\section{Introduction}\label{sec:intro} Clustering problems with centers are well
studied in the fields of operation research, machine
learning, data processing, theoretical computer science, and various other disciplines. Given a point-set $ X$ consisting of $n$  points, a distance metric $d$ defined on $X$, and a parameter $k$, the goal is to partition $X$ into $k$ clusters $C_1, C_2, \ldots, C_k$ with respective centers $c_1, c_2, \ldots, c_k$. The $k$-\textsc{Center}, $k$-\textsc{Median}, and  $k$-\textsc{Means} problems are notably the most widely studied clustering problems of this type:~\cite{stoc/CharikarP01,Gonzalez85,AryaGKMMP04,HsuN79,AgarwalP02}. The $k$-\textsc{Center} objective is to find a clustering such that the maximum distance between the center of a cluster and any point in the cluster, i.e., the maximum radius of a cluster, is minimized. The $k$-\textsc{Median} objective minimizes the total distance, which can average out larger individual costs, potentially overlooking the needs of outlier points.

The \MSR\ problem offers a compromise by targeting to minimize the sum of the radii of the clusters. We define the \MSR\ problem formally as follows.

\pbDef{Min-Sum-Radii (\MSR)}
{A finite set $X$ of $n$ points in a metric space $(X,d)$, an integer
  $k$, and a number $\cost \in \RR$.}
{Are there center-radius pairs $(c_1,r_1),\dots,(c_t,r_t)$ with $t\le k$,
  $c_i\in X$, and $r_i\ge 0$ such that every point $x\in X$ is
  of distance at most $r_i$ to some center $c_i$ (i.e., $d(x,c_i)\le
  r_i$) and $\sum_{i=1}^t r_i\leq \cost$.}

This approach allows for more balanced cluster formation, fine-tuning the radii across all clusters and ensuring that the maximum individual cost is kept within reasonable bounds,  
although it may be higher than in $k$-\textsc{Center}—potentially by a factor of $k$—though this is not always the case. This intermediate objective is particularly useful in applications where balancing the load among clusters is crucial, such as in wireless network design where the goal is to minimize the energy required for transmission, which is proportional to the sum of the radii of the coverage areas of the base stations, see Lev-Tov and Peleg~\cite{Lev-TovP05}. 

Charikar and Panigrahy \cite{stoc/CharikarP01} studied the \MSR\ problem and noted that it effectively addresses the ``dissection effect'' inherent in the $k$-center approach, where the uniform maximum radius assumption can lead to significant overlap among clusters, resulting in points that should belong together being split across different clusters. By minimizing the sum of the radii, \MSR\ reduces this dissection effect, promoting more coherent and practical clustering solutions.

Gibson et al.~\cite{swat/GibsonKKPV08} showed that the \MSR\ problem is \NP-hard, even when restricted to planar metric or metrics with constant doubling dimension. For general metrics, they presented a randomized algorithm for \MSR\ with a runtime of $n^{O(\log n \log D)}$ where $D$
represents the ratio between the maximum and minimum distances. They also obtained
a $(1 + \epsilon)$-approximation in quasi-polynomial time. It follows that under standard complexity-theoretic assumptions, the problem cannot be \APX-hard.  Interestingly, Gibson et al.~\cite{GibsonKKPV12} presented a polynomial time algorithm for \MSR\ in constant-dimensional Euclidean metrics. Charikar and Panigrahy~\cite{stoc/CharikarP01} developed a 
$O(1)$-approximation algorithm for the \MSR\ problem (as well as the $k$-\textsc{Min-Sum-Diameter} problem) on general metrics using the primal-dual framework introduced by Jain and Vazirani for the $k$-\textsc{Median} problem. 
This was further refined by Friggstad and Jamshidian~\cite{FriggstadJ22}, who achieved a $3.389$-approximation for \MSR. This was recently improved to $3+\varepsilon$, which is currently the best-known approximation factor for the general case~\cite{DBLP:conf/soda/BuchemERW24}.

The \MSR\ problem has been studied in constrained settings, such as
with lower bounds, outliers, and capacities. A significant line of
work focuses on its parameterized complexity. 
Several recent results~\cite{esa/0002V20, BandyapadhyayL023a, Jaiswal0Y24, aaai/ChenXXZ24, CartaDHR024, DBLP:conf/aaai/BanerjeeBGH25} discuss \FPT-approximation algorithms for the problem and its variants; see the related work section.

We investigate the \MSR\ problem under the \emph{graph metrics}. For an undirected graph $G=(V,E)$ with positive (or unit) edge weights,
we denote by $X=V$ the set of input points, and by $d(u,v)$, the length of a shortest $u$--$v$ path in $G$.
The metric space induced by the graph $G$ is $(X,d)$.

\pbDefP{\MSR\ on Graph Metrics}
{An undirected graph $G=(V,E)$ with positive edge weights, an integer
  $k$, and a number $\cost\in \RR$. The input
  metric $d$ is the graph metric induced by $G$.}
{$k$ (the (maximum) number of allowed clusters)}
{Are there center–radius pairs $(c_1,r_1),\dots,(c_t,r_t)$ with $t\le k$, $c_i\in V$, and $r_i\ge 0$ such that every vertex $v\in V$ is
  of distance at most $r_i$ from some center $c_i$ (i.e., $d(v,c_i)\le
  r_i$) and $\sum_{i=1}^t r_i\leq \cost$.}

    Throughout this paper, we only consider undirected graphs, and all of our results are stated for undirected graph metrics.
We distinguish between weighted graph metrics (with positive edge weights) and (unweighted) graph metrics (where all edges have unit weight).

\begin{table}[b]
\caption{Summary of complexity results for \MSR\ on weighted graph
  metrics. A checkmark (\cmark) indicates that the parameter in that
  column is used for the corresponding result. $k$: number of
  clusters; $\Delta$: cost of optimal clustering; $vc$: vertex cover
  number; $tw$: treewidth. Note that these results provide a
  comprehensive picture for the parameterized complexity of \MSR\ with
respect to any combination of these 4 parameters.}
\centering
\small
\begin{tabular}{|c|c|c|c|c|c|}
\hline
$k$ & $\cost$ & $\vc$ & $\tw$ & Complexity & Reference \\
\hline
\cmark & \cmark &  &  & \W[1] & ~\autoref{thm:MT} \\
\hline
\cmark &  & \cmark &  & \W[1] & ~\autoref{thm:vc+k} \\
\hline
\cmark &  &  &  & \XP & naive brute force \\
\hline
 &  &  & \cmark & \XP & Drexler et al.~\cite{DBLP:journals/corr/abs-2310-02130} \\
\hline
 & \cmark &  & \cmark & \FPT & ~\autoref{cor:tw+Delta} \\
\hline
\end{tabular}

\label{tab:results-summary}
\end{table}

\textbf{Weighted Graph Metrics.} 
 A naive brute-force algorithm yields an \XP\ algorithm parameterized
 by $k$: one can guess the $k$ centers (at most $\binom{n}{k}$
 choices) and then determine the optimal radii (for each radius there
 are at most $n-1$ possible choices--as there are at most $n-1$
 possible $r$-neigborhoods, i.e., neighborhoods around a vertex with
 radius $r$), leading to a running time of $n^{\bigoh(k)}$.
To the best of our knowledge, the question of whether \MSR\ admits a true \FPT\ algorithm in this setting--an algorithm running in time $f(k)\cdot n^{O(1)}$ for some computable function~$f$--has remained open.

We resolve this by proving that \MSR\ is \W[1]-hard parameterized by $k+\cost$, even on weighted bipartite graphs. Furthermore, we establish a lower bound under the Exponential Time Hypothesis (\ETH).

\begin{restatable}{theorem}{thmMT}\label{thm:MT}
    \MSR\ on metrics induced by  weighted bipartite graphs is \textup{\W[1]\text{-hard}} parameterized by the number of clusters $k$ plus the cost of the clustering $\cost$. Moreover, under the Exponential Time Hypothesis, there is no $f(k+\cost)\cdot N^{o(k+\log_2\cost)}$-time algorithm, where $N$ is the number of input points, that solves \MSR\ for metrics induced by weighted bipartite graphs, for any computable function $f$.
\end{restatable}

We note that when the inter-point distances are bounded by a polynomial in the input-size, the algorithm of Gibson et al.~\cite{swat/GibsonKKPV08} admits randomized \textit{quasi-polynomial} time. At this point, we mention that our \W[1]-hardness proof also implies \NP-hardness of the decision version of the problem and justifies the exponential edge-weights in the reduction of \autoref{thm:MT}.

We then explore the parameterized complexity of \MSR\ with respect to structural parameters. Drexler et al.~\cite{DBLP:journals/corr/abs-2310-02130} showed that the \MSR\ problem admits an \XP\ algorithm on metrics induced by bounded treewidth graphs, and asked the question whether the problem admits an \FPT\ algorithm   when parameterized by the treewidth $tw$. We first answer their question negatively, and more strongly show that \MSR\ problem stays \textup{\W[1]\text{-hard}} on metrics induced by  weighted graphs when parameterized by $vc$, the vertex cover number, plus $k$ (even on bipartite graphs). 

\begin{restatable}{theorem}{thmvwk}\label{thm:vc+k}
The \MSR\ problem is \textup{\W[1]\text{-hard}} on metrics induced by
weighted bipartite graphs when parameterized by the vertex
cover number ($\vc$) plus the number of clusters ($k$). Moreover, there is no $f(\vc+k)\cdot N^{o(\vc+k)}$-time algorithm for any computable function $f$ unless \ETH\ fails, where $N$ is the number of input points.
\end{restatable}

 We then, using the \XP\ algorithm of Drexler et al.\cite{DBLP:journals/corr/abs-2310-02130}, show that the \MSR\ problem is \FPT\ when parameterized by the combined parameter $tw$, the treewidth and $\cost$, the cost of clustering (\autoref{cor:tw+Delta}).

Taken together, these results already provide a rather comprehensive picture
of the parameterized complexity of \MSR\ on metrics induced by
 weighted graphs for combinations of the parameters $k$,
$\cost$ and any sparse structural parameter above and including the
vertex cover number (please refer to \autoref{tab:results-summary} for
an overview of our results). We complete the picture for dense
structural parameters (such as neighborhood diversity and modular
width) by showing that \MSR\ is already \NP-hard and \W[1]-hard
parameterized by $k+\cost$ on complete as well as complete bipartite graphs.

On the positive side, we obtain fixed-parameter tractability when
parameterizing by $\tw$ plus $\Delta$ (\autoref{cor:tw+Delta}).

\textbf{(Unweighted) Graph Metrics}
We now consider the case of metrics induced by (unweighted) graphs. Note
that as the  randomized  algorithm of Gibson et
al.~\cite{swat/GibsonKKPV08}, solves \MSR\ on undirected 
graphs in quasi polynomial time (as the maximum interpoint distance is
at most $n-1$), \MSR\ on these metrics should not be \NP-hard, unless
\NP\ is contained in randomized quasipolynomial time--an unlikely
complexity theoretic consequence. However, it is a long standing open question whether \MSR\ on undirected graph metrics, is solvable in polynomial-time; this is only known for the version of \MSR\ without zero clusters (i.e., no singleton clusters are allowed)~\cite{BehsazS12}.

Although we cannot answer
this question, we provide the following main contributions for these metrics. First, we
exhibit two natural variants of \MSR, namely, \exact\ and \fixed, and provide strong hardness
results for these variants even on metrics induced by undirected
graphs. In the
\exact\ version, it is required to find a solution with exactly
$k$ non-zero clusters. We formally define it as follows.

\pbDefP{\exact}
{An undirected graph $G=(V,E)$, an integer $k$, and a number $\cost\in
  \RR$. The input metric $(X,d)$ is the graph metric induced by $G$.}
{$k$ (the number of clusters)}
{Is there a set $C=\{(c_1,r_1),\dots,(c_k,r_k)\}$ of exactly $k$ center-radius pairs $(c_i,r_i)$ such that:
\begin{itemize}
\item $c_i \in X$ and there are at least two points in $X$ with
  distance at most $r_i$ from $c_i$,
\item every point in $X$ is within the distance $r_i$ from $c_i$ for at
  least one $i\in [k]$,
\item $\sum_{i \in [k]}r_i\leq \cost$.
\end{itemize}}

We show that \exact\ in
\NP-hard and \W[1]-hard parameterized by $k+\cost$ and provide a lower bound under \ETH\ (\autoref{thm:EXACT-MSR}). The exact-$k$ constraint and prohibition of singleton clusters in \exact\, are not artificial but reflect real-world needs in network partitioning, resource allocation, and community detection. In such settings, singleton clusters are undesirable for robustness and minimum-load considerations, and a fixed number of clusters is required for budget or regulatory constraints.

In the \fixed\
version, one is additionally required to choose the cluster centers
from a given set of allowed centers.

\pbDefP{\fixed}
{An undirected graph $G=(V,E)$, an integer $k$, a set of allowed
  centers $A\subseteq V$, and a number $\cost\in RR$. The input metric $(X,d)$ is the graph metric induced by $G$.}
{$k$}
{Is there a set $C=\{(c_1,r_1),\dots,(c_t,r_t)\}$ of $t\leq k$
  center-radius pairs $(c_i,r_i)$ such that:
\begin{itemize}
\item $c_i \in A$ and $r_i\geq 0$ is a number,
\item every point in $X$ is within distance $r_i$ from $c_i$ for at
  least one $i\in [t]$,
\item $\sum_{i \in [t]}r_i\leq \cost$.
\end{itemize}}

The restriction to allowed centers in \fixed\ models practical settings, where the centers can only be selected from a predefined set of feasible locations.
We show
that \fixed\ is \W[1]-hard parameterized by either $k+\cost$ (\autoref{thm:FixedMSR}) or $k+|A|+\fvs$ (\autoref{thm:FixedMSRfvs}),
where $A$ is the set of allowed centers and $\fvs$ is the \emph{feedback vertex set 
number} of the underlying undirected graph. It is important to note that we are not able to show \NP-hardness for \fixed\ and in fact we see this as a way around first having to settle the classical complexity of \MSR. That is, to obtain a parameterized complexity classification, it is sufficient to show \W[1]-hardness, which as our result for \fixed\ shows is an easier target.
 
Along with these results, we also
show that \MSR\ as well as \exact\ and \fixed\ are fixed-parameter
tractable parameterized by the \emph{treedepth} of the underlying graph (\autoref{thm: treedepth}).

\subsection{Further Related Work} In this subsection, we mention a few more related research works on \MSR\ problem. Bil{\`o} et al.~\cite{BiloCKK05} presented a polynomial time algorithm for \MSR\ when the input points are on a line. Proietti and Widmayer~\cite{isaac/ProiettiW06}, considered the problem of partitioning the nodes of a graph $G$ into exactly $k$ non-empty subsets, in order to minimize the sum of the induced subgraph radii. They show that the problem is \NP-hard when $k$ is part of the input. They presented an exact algorithm for the problem for a fixed constant $k > 2$ running in $O(n^{2k}/k!)$ time.

Clustering problems have been studied extensively from the perspective of exact algorithms. An initial significant result in this domain is an exact algorithm for the 
$k$-\textsc{Center} problem by Agarwal and Procopiuc~\cite{AgarwalP02}. They presented an $n^{O(k^{1 - 1/d})}$ % $n^{O(k^{ 1- \frac{1}{d} })}$ 
time algorithm in $\mathbb{R}^d$.
  Notably, in two-dimensional space, their algorithm runs in  $2^{O(
\sqrt{n \log n})}$
time for any value of $k$, i.e., in sub-exponential time. Fomin et al.~\cite{iwpec/FominG0P022} explored exact exponential algorithms for various clustering problems, including 
$k$-\textsc{Median} and 
$k$-\textsc{Means}. They developed a $O^*(1.89^n)$ time algorithm beating the trivial $O(2^n)\cdot \poly(n)$ time for these problems. For the facility location version of  
$k$-\textsc{Median} and $k$-\textsc{Means} problems, they presented a 
$2^n\cdot\poly(n)$ algorithm using a novel technique of \textit{subset convolution} studied by Bj{\"o}rklund  et al.~\cite{BjorklundHK09}. A detailed survey on the techniques for designing exact exponential algorithms can be found in the text by Fomin and Kratsch~\cite{FominK10}.

Inamdar and Varadarajan~  \cite{esa/0002V20} derived a 28-approximation for the uniformly capacitated \MSR\ problem with a running time of 
$O(2^{O(k^2)}\cdot n^{O(1)})$. Bandyapadhyay et al.~\cite{BandyapadhyayL023a} also provided \FPT-approximations: they developed a $(4 + \epsilon)$-approximation algorithm running in time of $2^{O(k \log(k/\epsilon))}\cdot n^3$ for \MSR\ with uniform capacities, and a $(15 + \epsilon)$-approximation algorithm for \MSR\ with non-uniform capacities, which runs in $2^{O(k^2\log k)}\cdot n^3$ time.  Interestingly, they also proved the \W[2]-hardness of a capacitated version of \MSR.
    Recently, Jaiswal et al.~\cite{Jaiswal0Y24} have improved upon these results. They provided a  $(3 + \epsilon)$-approximation algorithm with a running time of $2^{O(k \log(k/\epsilon))}\cdot n^3$ for \MSR\ with uniform capacities, a $(4 + \sqrt{13})$-approximation algorithm for \MSR\ with non-uniform capacities, which runs in $2^{(k^3+k\log(k/\epsilon))}\cdot \poly(n)$ time. Further, their results extend to a constant
factor approximation ratio for any $L_p$ norm, where $p \geq 1$. Filtser and Gadekar~\cite{DBLP:journals/corr/abs-2409-04984} obtained a $(3 + 2\sqrt{2} + 
\varepsilon)$-approximation for \MSR\ with non-uniform capacities, 
running in time $2^{O(k^2 \log k + k \log(k/\varepsilon))} \cdot 
n^{O(1)}$, which also extends to the case where the objective is a 
monotone symmetric norm of the radii, as well as to cluster capacities.
More recently, Chen et al.~\cite{aaai/ChenXXZ24} presented an \FPT-$(2+\epsilon)$-approximation algorithm for the \MSR\ problem which runs in time $2^{O(k\log (k/\epsilon))}n^3$.
They also considered variants of \MSR\ with capacity and matroid constraints. For both constraints, they obtained a $(3+\epsilon)$-approximation algorithm that runs in time $2^{O(k\log(k/\epsilon))}n^3$. 
Recently, Carta et al.~\cite{CartaDHR024} derived constant-factor-approximations for \MSR\ with a fairness constraint. Banerjee et al.~\cite{DBLP:conf/aaai/BanerjeeBGH25} presented an exact algorithm for Min-Sum-Diameters with running time $n^{O(k)}$, along with $(1+\epsilon)$-approximation algorithms for both Min-Sum-Radii and Min-Sum-Diameters in metric spaces of bounded doubling dimension~$d$, achieving runtime $O(kn) + (1/\epsilon)^{O(dk)}$, with extensions to fairness constraints and mergeable clustering. Their work also includes a reduction from the \textsc{Grid-Tiling} problem (Theorem 33 in the \mbox{arXiv}
 version of ~\cite{aaai/ChenXXZ24}); although not explicitly stated, the reduction is a parameterized reduction. Since \textsc{Grid-Tiling} is known to be \W[1]-hard~\cite{marx2007optimality, CyganFKLMPPS15}, this establishes the \W[1]-hardness of \MSR\ parameterized by $k$ and $\cost$ the cost of the clustering in constant doubling dimension metrics.

Drexler et al.~\cite{DBLP:journals/corr/abs-2310-02130} studied a facility variant of \MSR\ problem called \textsc{Msrdc} ( Minimum Sum of Radius Dependent Costs ) where we are given $k\in \mathbb{N}$, a metric space $(V,d)$ where $V=F\cup C$ for facilities $F$ and clients $C$. The objective is to find a clustering given by $k$ facility-radius pairs $(f_i,r_i)$ with a goal to minimize the sum of cluster costs $\sum_{i=1}^k r_i^{\alpha}$ for some $\alpha\in \mathbb{R}_{\geq 1}$. They devised an $\XP$ algorithm for the problem on metrics induced by bounded treewidth graphs. Note that \MSR\ is a special case of \textsc{Msrdc} with $\alpha$ set to 1.
They also showed that even if we know the set $C\subset X$ of the centers, \MSR\ remains \NP-hard even on metrics induced by weighted planar graphs and bounded doubling dimension.

\section{Preliminaries} 
We use the standard graph-theoretic notation as in Diestel \cite{Diestel12}.
We consider graphs that are undirected, connected, finite and simple. We denote by $G=(V,E)$ a graph with vertex set $V(G)$ and edge set $E(G)$. We denote the edge between $u$ and $v$ by $\{u,v\}$. We use $n$ to denote the number of vertices of $G$. For a set $X \subseteq V$, the graph $G[X]$ denotes the induced subgraph of $G$ on the vertex set $X$. The open neighborhood of a vertex $v$, denoted by $N(v)$, is the set of vertices adjacent to $v$ and the closed neighborhood of $v$ is denoted by $N[v]=N(v) \cup \{v\}$. We write $[n] = \{1,2,\dots,n\}$ for a positive integer $n$. A graph $G$ is a \textit{clique} if for every pair of vertices $u$ and $v$ there is an edge $\{u,v\}$ in $E(G)$.
A \emph{path} $P$ on the $\ell\leq |V(G)|$ vertices in $G$ is a sequence of vertices
$P = (v_1, v_2, \ldots, v_\ell)$ such that $\{v_{i-1}, v_i\} \in E$ for every
$i \in \{2,\ldots,\ell\}$. If $G$ is an edge-weighted graph with weight function $w : E \to \mathbb{R}_{\ge 0}$,
the \emph{weight} of $P$ on $\ell$ vertices, is defined as
$w(P) = \sum_{i=2}^{\ell} w(\{v_{i-1}, v_i\})$.
If $G$ is unweighted, the \emph{weight} (or \emph{length}) of $P$ on $\ell$ vertices is simply the
number of edges in $P$, that is, $w(P) = \ell-1$.
For two vertices $u,v \in V(G)$, the \emph{distance} between $u$ and $v$, denoted by $d(u,v)$, is defined as
$d(u,v) = \min\{ w(P) \mid P \text{ is a } u\text{--}v \text{ path in } G \}$.
If no $u$–$v$ path exists, we define $d(u,v) = \infty$.

  \textbf{Parameterized Complexity}:
 In Parameterized Complexity theory, the computational complexity of a problem is measured as a function of the input size $|I|$ and a non-negative integer parameter $k$ associated with the input.
We say a parameterized problem is \emph{fixed-parameter tractable} (\FPT) with respect to a parameter $k$, if 
there exists an algorithm that runs in time 
$f(k)\cdot |I|^{O(1)}$ where $f$ is a computable function independent of the input size $|I|$ and $k$ is a parameter associated with the input instance.

Let $P, Q \subseteq \Sigma^* \times \mathbb{N}$ be two parameterized problems.  A \emph{parameterized reduction} from $P$ to $Q$ is an algorithm $\mathcal{A}$ that, given an instance $(I, k)$ of $P$, outputs an instance $(I', k')$ of $Q$ such that $(I, k)$ is a \textsf{YES}  instance of $P$ if and only if $(I', k')$ is a \textsf{YES}  instance of $Q$,
$k' \leq g(k)$ for some computable function $g \colon \mathbb{N} \to \mathbb{N}$, and
the running time of $\mathcal{A}$ is $f(k)\cdot |I|^{O(1)}$ for some computable function $f$.
The hardness theory in parameterized complexity is developed via the notion of parameterized reductions and a hierarchy of complexity classes forming the \emph{W}-hierarchy: $\FPT \subseteq \W[1] \subseteq \W[2] \subseteq \cdots \subseteq \XP$.

In a $q$-CNF-SAT instance, we are given a $q$-CNF formula $\Phi$, and the
question is to decide whether $\Phi$ is satisfiable. \cite{ImpagliazzoP01}, in their seminal work, formulated the
following hypothesis, called Exponential Time Hypothesis(\ETH) which states that $q$-CNF-SAT, $q\geq  3$ cannot be solved within a
running time of $2^{o(p)}$ or $2^{o(m)}$, where $p$ is the number of variables and $m$ is the number of clauses in the input $q$-CNF formula. The Exponential Time Hypothesis ($\ETH$) is known to imply $\FPT\neq \W[1]$. We rely on the definitions of the structural graph parameters such as vertex cover number, treedepth, feedback vertex set, treewidth, cliquewidth, neighborhood diversity, modular width, and modular treewidth etc., for more details on the subject, the reader is deferred to \cite{CyganFKLMPPS15}.

\section{\MSR\ on metrics induced by weighted graphs}

In this section, we present our results for \MSR\ on metrics induced
by weighted graphs. In particular, we show in \Cref{ssec:wkcost} that
\MSR\ on metrics induced by weighted bipartite graphs is \W[1]-hard
parameterized by $k+\cost$. We then continue in \Cref{ssec:wkvc} to
show that the same result applies when parameterized by $k+\vc$, where
$\vc$ is the vertex cover number of the graph. Interestingly,
the latter result even applies for the variant of \MSR, where the set of centers is given and
one merely needs to find the right radius for every center, i.e., the
hardness of the problem lies in finding the right radius for every
center. We then complete the picture in \Cref{ssec:wcomplete} by
presenting further hardness results for dense graph parameters and an
algorithm using the parameter treewidth. Together, those results (which
are also given in \Cref{tab:results-summary} apart from the results for dense graph
parameters) provide a comprehensive picture of the parameterized
complexity of \MSR\ with respect to any combination of the parameters
$k$, $\cost$, structural parameters less restrictive than vertex
cover, and common structural parameters for dense graphs (such as
neighborhood diversity and modular width).

\subsection{Parameterized by $k+\Delta$}\label{ssec:wkcost}

Here, we show that \MSR\ on metrics induced by weighted bipartite
graphs is \W[1]-hard parameterized by $k+\cost$. We prove the theorem by giving a parameterized reduction from the well-known $\W[1]$-hard problem \mcc.
The reduction carefully encodes the selection of a multicolored clique through weighted clustering constraints. 

\thmMT*

\begin{proof}
 Let $(G, k, (V_1, \ldots , V_k))$ denote an instance of the \mcc\
 problem. We create an instance $(X,d,k,\cost)$ of the \MSR\ problem,
 where $X$ is the set of vertices of the weighted bipartite graph $G'$
 and $d$ is the shortest-path metric induced by $G'$, using the following steps (see also \autoref{fig:example:construction:thm1} for
 an illustration of the construction): 
\begin{itemize}
\item $G'$ has one vertex $v$ for every vertex $v \in V(G)$.

\item For each $i\in [k]$, add a vertex $x_i$ and make $x_i$ adjacent
  to every vertex in $V_i$, with each edge having a weight of
  $2^{i-1}$.
\item For each $1 \leq i\leq k$, introduce a set $W^i$ consisting of $k+1$ vertices, and make each $w\in W^i$ adjacent to every vertex in $V_i$, with each edge having a weight of $2^i$.

\item For every non-edge $e \notin E(G)$ with endpoints $u \in V_i$ and $v \in V_j$, introduce a vertex $w_e$ in $G'$ and connect it to every vertex in $(V_i \cup V_j) \setminus \{u, v\}$. Assign a weight of $2^i$ to every edge $\{w_e,a\}$ for $a \in V_i$ where $a\neq u$, and a weight of $2^j$ to every edge $\{w_e,b\}$ for $b \in V_j$ where $b\neq v$. Let $\overline{E}_{i,j}$ denote the set of these vertices. 

\item Let  
$X=V'(G')$ be the point set in the metric space, and let $d$ denote the shortest path metric in $G'$. The terms vertex
and point are used interchangeably. 
\item Finally, set the cost $\cost$ to $2^{k+1}-2$
\end{itemize}
This completes the construction of the instance $(X,d,k,\cost)$ of the \MSR\ problem. We observe that since the graphs $G'[\cup_{i\in [k]} V_i]$ and $G'- (\cup_{i\in [k]} V_i)$ are independent sets, the obtained graph $G'$ is bipartite. Moreover, the construction takes $\poly(|V(G)|,k)$-time i.e., polynomial time.

    \begin{intuitionbox}
The high-level idea of the reduction is the following. The construction relies on exponentially increasing edge weights and enforces a correspondence between selected centers and the vertices of a multicolored clique in the input graph $G$ of \mcc\ instance.
The natural candidates for center locations are the partition classes $V_i$. To force the selection of one center from each $V_i$ for every $i \in [k]$, we introduce a vertex-set $W^i$ of size $k+1$. Since at most $k$ centers can be selected, at least one vertex $w \in W^i$ is not chosen as a center. Covering such a vertex within the budget $\cost$ is possible only if a center lies in $V_i$, thereby forcing the choice inside each part.
Each $V_i$ is kept independent to preserve bipartiteness. We add an apex vertex $x_i$ adjacent to all vertices of $V_i$ so that a chosen center $v_i \in V_i$ covers its entire part with radius exactly $2^i$. This radius is tight: any smaller radius fails to cover even $V_i$, while any deviation from this scale incurs a cost exceeding $\cost$ due to the exponential separation of weights.
Finally, if the selected centers do not correspond to a multicolored clique in $G$, the exponential weight scheme forces a higher-layer penalty, again violating the budget. Thus, a feasible clustering corresponds precisely to a multicolored clique in $G$.
\end{intuitionbox}

\begin{figure} [t] 
\centering
\begin{tikzpicture}[mynode/.style={circle,fill=blue, inner sep=0pt, minimum width=0.2cm}, 
mb/.style={mynode,draw=blue,fill=blue},
mr/.style={mynode,draw=red,fill=red},
mg/.style={mynode,draw=green!50!black,fill=green!50!black},
mbl/.style={mynode,draw=black,fill=black},
sn/.style={minimum width=0.1cm}]

  %---------------- Left: G ----------------%
  \node (b1) [mb,label=left:{$b_1$}] at (0,1) {};
  \node (b2) [mb,label=left:{$b_2$}] at (0,0) {};
  \node (r1) [mr,label=below right:{$r_1$}] at (2,1) {};
  \node (r2) [mr,label=left:{$r_2$}] at (2,0) {};
  \node (r3) [mr,label=right:{$r_3$}] at (2,-1) {};
  \node (g1) [mg,label=above:{$g_1$}] at (4,1) {};

  \draw [thick,black] (b1) edge (r1);
  \draw [thick,black] (r1) edge (g1);
  \draw [thick,black] (r1) edge (b2);
  \draw [thick,black] (b2) edge (r3);
  \draw [thick,black, bend left=25] (b1) edge (g1);
  \draw [thick,black] (r3) edge (r2);
  \draw [thick,black] (r2) edge (r1);
  \draw [thick,black] (r2) edge (g1);

  %---------------- Right: G' ----------------%
  \begin{scope}[xshift=8cm]
  \begin{scope}[xshift=-1.5cm]

  % Ellipsoids
  \draw[fill=blue!30]  (0, 0.5) ellipse [x radius=0.2, y radius=1];
  \draw[fill=red!30]  (2.5,0) ellipse [x radius=0.2, y radius=1.25];
  \draw[fill=green!50!black!30!white] (5, 1) ellipse [x radius=0.2, y radius=0.2];

  % Main vertices
  \node (b1k) [mb,label=left:{$b_1$},label=right:{\textcolor{blue}{$2^1$}}] at (0,1) {}; 
  \node (b2k) [mb,label=left:{$b_2$}] at (0,0) {}; 
  %\draw (b1k) -- (b2k);

  \node (r1k) [mr,label=above:{$r_1$},label=right:{\textcolor{red}{$2^2$}}] at (2.5,1) {}; 
  \node (r2k) [mr,label=left:{$r_2$}] at (2.5,0) {};
  \node (r3k) [mr,label=left:{$r_3$}] at (2.5,-1) {}; 
  %\draw  (r1k) to (r2k) to (r3k); 
  %\draw [bend left=30] (r1k) edge (r3k);

  \node (g1k) [mg,label=above:{$g_1$},label=right:{\textcolor{green!50!black}{$2^3$}}] at (5,1) {};

  %---------------- x_i vertices (LEFT of V_i blocks) ----------------%
  % x1 for V1 (blue)
  \node (x1) [mb,sn,label=left:{$x_1$},label=above:{\scriptsize \textcolor{blue}{$2^0$}}] at (-1.2,0.5) {};
  \draw [black,opacity=0.5] (x1) -- (b1k);
  \draw [black,opacity=0.5] (x1) -- (b2k);

  % x2 for V2 (red)
  \node (x2) [mr,sn,label=right:{$x_2$},label=above:{\scriptsize \textcolor{red}{$2^1$}}] at (3.8,0) {};
  \draw [black,opacity=0.5] (x2) -- (r1k);
  \draw [black,opacity=0.5] (x2) -- (r2k);
  \draw [black,opacity=0.5] (x2) -- (r3k);

  % x3 for V3 (green)
  \node (x3) [mg,sn,label=right:{$x_3$},label=above:{\scriptsize \textcolor{green!50!black}{$2^2$}}] at (6.3,0) {};
  \draw [black,opacity=0.5] (x3) -- (g1k);

  %---------------- W^1 leaves ----------------%
  \node (w11) [mb,sn,label=below:{$w_1^1$}] at (-0.75,-2) {};
  \node (w21) [mb,sn,label=below:{$w_2^1$}] at (-0.25,-2) {};
  \node (w31) [mb,sn,label=below:{$w_3^1$}] at (0.25,-2) {};
  \node (w41) [mb,sn,label=below:{$w_4^1$}] at (0.75,-2) {};
  \draw [black,opacity=0.25] (w11) -- (b2k)--(w21);  
  \draw [black,opacity=0.5] (w31) -- (b2k)--(w41);
  \draw [black,opacity=0.25, bend left=50] (b1k) to (w11);
  \draw [black,opacity=0.25, bend left=40] (b1k) to (w21);
  \draw [black,opacity=0.25, bend left=30] (b1k) to (w31);
  \draw [black,opacity=0.25, bend left=20] (b1k) to (w41);
  \end{scope}

  %---------------- W^2 leaves ----------------%
  \begin{scope}[xshift=1cm]
  \node (w12) [mr,sn,label=below:{$w_1^2$}] at (-0.75,-2) {};
  \node (w22) [mr,sn,label=below:{$w_2^2$}] at (-0.25,-2) {};
  \node (w32) [mr,sn,label=below:{$w_3^2$}] at (0.25,-2) {};
  \node (w42) [mr,sn,label=below:{$w_4^2$}] at (0.75,-2) {};
  \draw [black,opacity=0.25] (w12) -- (r3k)--(w22);  
  \draw [black,opacity=0.25] (w32) -- (r3k)--(w42);
  \draw [black,opacity=0.25, bend left=60] (r1k) to (w12);
  \draw [black,opacity=0.25, bend left=50] (r1k) to (w22);
  \draw [black,opacity=0.25, bend left=40] (r1k) to (w32);
  \draw [black,opacity=0.25, bend left=30] (r1k) to (w42);
  \draw [black,opacity=0.25, bend left=50] (r2k) to (w12);
  \draw [black,opacity=0.25, bend left=40] (r2k) to (w22);
  \draw [black,opacity=0.25, bend left=30] (r2k) to (w32);
  \draw [black,opacity=0.25, bend left=20] (r2k) to (w42);
  \end{scope}

  %---------------- Pair gadgets ----------------%
  \node (wb1r2) [mbl,label=above:{$w_{\{b_1,r_2\}}$}] at (-3, 2) {};
  \draw [black,opacity=0.25] (wb1r2) to (b2k);
  \draw [black,opacity=0.25] (wb1r2) to (r1k);
  \draw [black,opacity=0.25] (wb1r2) to (r3k);

  \node (wb1r3) [mbl,label=above:{$w_{\{b_1,r_3\}}$}] at (-1.5, 2) {};
  \draw [black,bend right=53,opacity=0.25] (wb1r3) to (b2k);
  
  %\draw [bend left=30] (r1k) edge (r3k);
  \draw [black,opacity=0.25] (wb1r3) to (r1k);
  \draw [black,opacity=0.25] (wb1r3) to (r2k);

  \node (wb2r2) [mbl,label=above:{$w_{\{b_2,r_2\}}$}] at (0, 2) {};
  \draw [black,opacity=0.25] (wb2r2) to (b1k);
  \draw [black,opacity=0.25] (wb2r2) to (r1k);
  \draw [black,opacity=0.25] (wb2r2) to (r3k);

  %---------------- W^3 leaves ----------------%
  \begin{scope}[xshift=3.5cm]
  \node (w13) [mg,sn,label=below:{$w_1^3$}] at (-0.75,-2) {};
  \node (w23) [mg,sn,label=below:{$w_2^3$}] at (-0.25,-2) {};
  \node (w33) [mg,sn,label=below:{$w_3^3$}] at (0.25,-2) {};
  \node (w43) [mg,sn,label=below:{$w_4^3$}] at (0.75,-2) {};
  \draw [black,opacity=0.25] (w13) -- (g1k)--(w23);  
  \draw [black,opacity=0.25] (w33) -- (g1k)--(w43);
  \end{scope}

  \node (wb2g1) [mbl,label=above:{$w_{\{b_2,g_1\}}$}] at (2, 2) {};
  \draw [black,opacity=0.25] (wb2g1) to (b1k);

  \node (wr3g1) [mbl,label=above:{$w_{\{r_3,g_1\}}$}] at (4, 2) {};
  \draw [black,opacity=0.25] (wr3g1) to (r1k);
  \draw [black,opacity=0.25] (wr3g1) to (r2k);

  \end{scope}
\end{tikzpicture}

\caption{On the left we see an example input graph $G$ for the \mcc\ problem with $V_1=\{b_1,b_2\},V_2=\{r_1,r_2,r_3\},V_3=\{g_1\}$ where $\{b_1,r_1,g_1\}$ forms a multicolored clique. On the right, we see graph $G'$. For each $i\in[3]$, a vertex $x_i$ is added and connected to every vertex in $V_i$ by edges of weight $2^{i-1}$. Except those edges having $x_i$ as an endpoint, all edges within the blue ellipsoid or leaving the blue ellipsoid have weight $2$, those in/leaving red have weight $4$, and those leaving $g_1$ (green) have cost $8$.
\label{fig:example:construction:thm1}}
\end{figure}
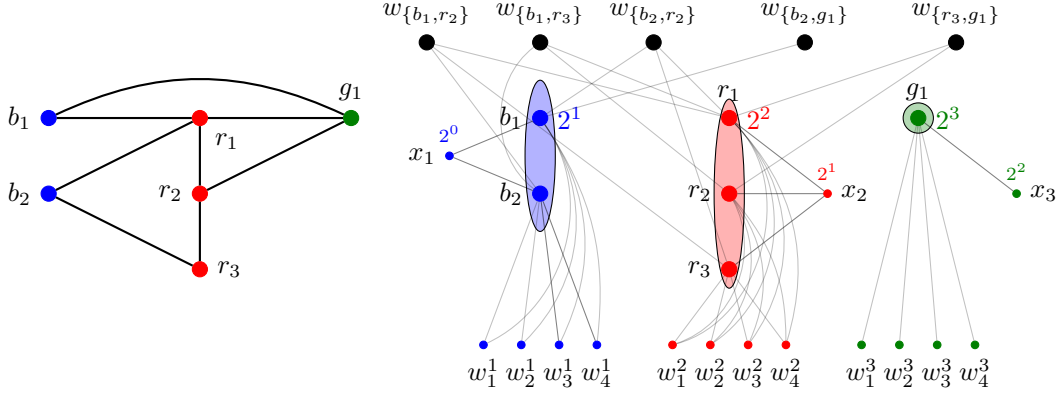

We show that $G$ has a multicolored clique of size $k$ if and only if $(X,d,k,\cost)$ has a \MSR\ clustering with at most $k$ clusters of cost at most $\cost$. For the forward direction, consider a \textsf{YES} instance $(G, k, (V_1,\ldots,V_k))$ of \mcc\ problem. Let $S=\{v_1,v_2,\ldots,v_k\}$ be a multicolored clique of size $k$ in $(G, k, (V_1, \ldots , V_k))$. We obtain a corresponding \MSR\ clustering $\mathcal{C}$ for $(X,d,k,\cost)$ as follows: for each  $v_i\in S$, we define the cluster $C_i$ with  center $v_i$ to contain the vertices within a distance of $2^i$ from $v_i$, i.e., $C_i=\{u: d(v_i,u)\leq 2^i \text{ under the metric } d\}$. 

Note that the number of clusters in $\mathcal{C}$ is $k$ and the cost $\cost$ of $\mathcal{C}$ is $2^1+2^2+\ldots+2^k=2^{k+1}-2$.
In the following, we show that $\mathcal{C}$ is a valid clustering, i.e., for every vertex $v \in X=V(G')$, there is a cluster in $\CCC$
  containing $v$.
\begin{itemize}

    \item For each $i\in [k]$, $x_i\in C_i$ because the edge $\{v_i,x_i\}$ has weight $2^{i-1}$ which is smaller than $2^i$, the radius of the cluster $C_i$.
    
    \item For each $i\in [k]$, $V_i\subseteq C_i$ because $V_i$ is an independent set and each vertex $v\in V_i$ can be reached by the center $v_i$ with a length two path $P$ via $x_i$. Hence, such a path has length $2\cdot 2^{i-1}=2^{i}$  which is equal to the radius of $C_i$.
    
    \item $W^i\subseteq C_i$ because $v_i$ is adjacent to each vertex in $W^i$ and each edge $\{v_i,w\}$ for a vertex $w\in W^i$, is of weight $2^{i}$, which is equal to the radius of $C_i$.

    \item $\overline{E}_{i,j}$ is covered in $\mathcal{C}$. Consider a vertex $w_e\in \overline{E}_{i,j}$ with $e=\{u,v\}$ where $u\in V_i$ and $v\in V_j$. As $u$ and $v$ are non-adjacent in $G$, we cannot have $u,v\in S$ because $S$ is a clique in $G$. Suppose $u\notin S$, then the chosen center $v_i\in S\cap V_i$ satisfies $v_i\neq u$, and by construction $w_e$ is adjacent to $v_i$ with weight $2^i$ which is exactly the radius of the cluster $C_i$. The symmetric argument holds for the case when $v\notin S$. Therefore, at least one of the cluster $C_i$ or $C_j$ contains $w_e$ and we assign $w_e$ to one of them arbitrarily if both are valid.  
\end{itemize}

For the reverse direction,  consider a \textsf{YES} instance $(X,d,k,\cost)$ of \MSR\ and let $\mathcal{C'}=(C_1,C_2,\ldots,C_k)$ be a clustering of $(X,d,k,\cost)$.
In the following, we observe that in any feasible clustering $\mathcal{C'}$ of $X$ with cost $\cost$, the centers lie inside the sets $V_i$.   

\begin{claim}\label{MainClm}
    Any valid clustering $\mathcal{C'}$ of $(X,d,k,\cost)$  contains, for each $1\leq i\leq {k}$, a cluster of radius at least $2^i$ centered at a vertex $v\in V_i$. 
\end{claim}

\begin{proof} Consider any clustering $\mathcal{C}'=\{C_1,C_2,\ldots,C_k\}$ of $X=V(G')$ with a cost at most $\cost=2^{k+1}-2$.
Let $\ell$ be the largest index such that there is no cluster in $\mathcal{C}'$ centered at a vertex $v\in V_{\ell}$ for $1\leq \ell \leq k$ with a radius at least $2^{\ell}$. Thus, for each $\ell+1\leq j\leq k$, there exists a cluster $C_j\in \mathcal{C}'$ centered at a vertex $v\in V_j$ with a radius at least $2^j$. Since $W^{\ell}$ has $k+1$ vertices, by the pigeonhole principle there is a vertex $w\in W^{\ell}$ such that no cluster in $\mathcal{C}'$ is centered at $w$. Let $C\in \mathcal{C}'$ be a cluster that contains $w$. We consider the following two cases:

\emph{\textbf{Case 1: Suppose $C=C_j$ for some $\ell+1\leq j \leq k$.}} We observe that the cluster $C$ contains $w\in W^{\ell}$ and by the choice of index $j$, the center of cluster $C$ is in $V_j$. Therefore, the radius of $C$ must be at least $2^j+2\cdot 2^{\ell}$. Summing the radii of the clusters indexed from $j=\ell+1$ to $k$, we obtain that the total cost of the clustering is at least $2^k+2^{k-1}+\ldots + 2^{\ell+1}+ 2\cdot 2^{\ell}$. Let $S=2^k+2^{k-1}+\ldots + 2^{\ell+1}+ 2^{\ell}$, then $2S-S= 2^{k+1}-2^{\ell}$. Thus the cost of the clustering $\mathcal{C}'$ is at least $S+2^{\ell}=2^{k+1}>2^{k+1}-2$.

\emph{\textbf{Case 2: Suppose $C\neq C_j$ for any $\ell+1\leq j \leq k$.}} We observe that any cluster containing $w \in W^{\ell}$ where $w$ is not a center has a radius of at least $2^{\ell}$ because $w$ is only adjacent to points in $V_{\ell}$ via edges of weight $2^{\ell}$.  Moreover, by the definition of the index $\ell$, the center of the cluster $C$ is not in $V_{\ell}$ and hence the radius of the cluster $C$ is at least $2^{\ell-1}+2^{\ell}$. Furthermore, we can assume that $\ell>1$, as otherwise if $\ell=1$ then the cost of the clustering $\mathcal{C}'$ would be at least $2^0+2^1+2^2+\ldots+ 2^k>\cost$. 

Therefore, as $\ell>1$ and $|W^{\ell-1}|=k+1$, by the pigeonhole principle, there exists $w'\in W^{\ell-1}$ which is not the center of any cluster in $\mathcal{C}'$. %Let $C'$  be the cluster that contains $w'$. 
Since all incident edges on $w'$ have cost $2^{\ell-1}$, and the cluster $C$ with radius $2^{\ell-1}+2^{\ell}$ cannot cover $w'$, in addition to the sum of the radii of the clusters indexed at $j$ for $\ell+1\leq j\leq k$, to cover $w$ (by cluster $C$) and $w'$ (by cluster $C$ or otherwise), the additional radius of at least $2^{\ell-1}+2^{\ell}+2^{\ell-1}$ must be added. Hence, the cost of clustering is at least $2^k+2^{k-1}+\ldots + 2^{\ell+1}+2^{\ell-1}+2^{\ell}+2^{\ell-1}= 2^{k+1}>2^{k+1}-2$.

Thus, in both cases, we get a contradiction to the assumption that the cost of the clustering  $\mathcal{C}'$ is at most $\cost=2^{k+1}-2$.
\end{proof}

By \autoref{MainClm}, the clustering $\mathcal{C'}$ with a cost at most $\cost=2^{k+1}-2$ contains clusters $C_1,C_2,\ldots,C_k$, where each cluster $C_i$ centered at a vertex $v\in V_i$ and of radius exactly equal to $2^i$. 

\begin{claim} Let $\mathcal{C'}$ be a \MSR\ clustering of $(X,d,k,\cost)$ of cost at most $\cost=2^{k+1} - 2$. Let $S' = \{v_1, v_2, \ldots, v_k\}$ be the set of cluster centers in $\mathcal{C'}$. Then $S'$ forms a multicolored clique of size $k$ in $G$.
\end{claim}
 
\begin{proof} For the sake of contradiction, suppose $e=\{v_i,v_j\}$ where $v_i,v_j\in S$ for $1\leq i\neq j\leq k$, is a non-edge in $E(G)$. Thus by construction, $w_e$ is not adjacent to $v_i$ and $v_j$ in $G'$. This implies that to cover $w_e$, one of the clusters in $\mathcal{C}'$ centered at $v_i\in V_i$ or $v_j\in V_j$ must have radius larger than $2^i$ or $2^j$. Consequently, by \autoref{MainClm} the cost of clustering $\mathcal{C'}$ is more than $\cost=2^{k+1}-2$, a contradiction to the total cost of the clustering $\mathcal{C'}$. Similarly, if there is another cluster $C_h\in \mathcal{C}'$ centered at $v_h\in S'$ for $1\leq h\neq i,j\leq k$ that covers $w_e$, its radius would be larger than $2^h$, which contradicts the total cost of the clustering $\mathcal{C'}$ due to \autoref{MainClm}. Therefore, the set $S'=\{v_1,v_2,\ldots,v_k\}$ is a multicolored clique of size $k$ in $G$ and thus the instance $(G, k, (V_1, \ldots , V_k))$ is a \textsf{YES} instance of \textsc{Multicolored Clique}.
\end{proof}

\mcc\ is \W[1]-hard when parameterized by $k$ the size of multicolored clique (Fellows et al.,~\cite{FellowsHRV09}, Pietrzak~\cite{Pietrzak03}). Thus as in our reduction, $k$ the size of multicolored clique is equal to the number of clusters in $\mathcal{C}'$, and the cost $\cost$ is upper bounded by a function of $k$, e.g., $g(k)=2^{k+1}$, we get that $\MSR$ is \W[1]-hard when parameterized by the combined parameter $k$, the number of clusters and $\cost$, the cost of the clustering.

Due to Corollary 14.23 in Cygan et al.~\cite{CyganFKLMPPS15}, we know that under $\ETH$, there is no $f(k)\cdot n^{o(k)}$-time algorithm for \mcc\ for any computable function $f$. As the reduction takes polynomial time, and the parameter $k+\cost=k+2^{k+1}-2$ grows exponentially with respect to the size of multicolored clique in $G$, we conclude that there is no $f(k)\cdot N^{o(k+\log_2\cost)}$-time algorithm for \MSR\ for any computable function $f$ under $\ETH$, where $N$ is the number of input points.
 \end{proof}

\subsection{Employing Structural Parameters for Sparse Graphs}\label{ssec:wkvc}

In this subsection, we show that \MSR\ on undirected weighted
bipartite graphs is \W[1]-hard parameterized by $k+\vc$, where $\vc$
is the vertex cover number of the graph.
We will use a
reduction from the \mcc\ problem, which is well-known to be \W[1]-hard
parameterized by the size of the multicolored clique~\cite[Lemma 1]{FellowsHRV09}.

\begin{definition}[\mcc]
  Given an undirected graph $G$, an integer $k$, and a partition
  $(V_1,V_2,\ldots,V_k)$ of the vertices of $G$; the task is to decide
  whether there is a $k$-clique containing exactly one vertex from each set
  $V_i$.
\end{definition}

\thmvwk*

\begin{proof}
  We provide a polynomial-time parameterized reduction from the \mcc\ problem. Since the parameter in our reduction 
  only grows linearly, additionally to \W[1]-hardness, we also obtain that there is no
  $f(vc+k)\cdot N^{o(vc+k)}$-time algorithm for \MSR\ under $\ETH$, where $N$ denotes the number of input points in the \MSR\ instance. 
  
    Let
  $(G, k, (V_1, \ldots , V_k))$ denote an instance of the \mcc\
  problem. Without loss of generality, we can assume that
  $V_i=\{v^i_1, \dotsc, v^i_n\}$ for every $i \in [k]$. We create an
  instance $(X,d,k',\cost)$ of the \MSR\ problem, where $k'=2k$, 
  $X$ is the set of vertices of the weighted undirected graph $G'$ and
  $d$ is the metric induced by $G'$ as follows; see also \Cref{fig:gadget-eij} for
  an illustration of the construction. The graph $G'$ contains the
  following vertices:
  \begin{itemize}
  \item one vertex $v$ for every $v \in V(G)$,
  \item two vertices \low\ and \high\ for every $i \in [k]$,
  \item two sets $W^i_{\lowl}$ and $W^i_{\highl}$ of
    $2k+1$ vertices each for every $i \in [k]$,
  \item a vertex $w_e$ for every non-edge $e =\{u,v\}\notin E(G)$ with endpoints $u
    \in V_i$ and $v \in V_j$ where $1\leq i< j\leq k$.
    Let $\overline{E}_{i,j}$ denote the set of these vertices.
  \end{itemize}

  \newcommand{\wfl}{\omega_{\lowl}}
  \newcommand{\wfh}{\omega_{\highl}}
Moreover, $G'$ has the following edges:
  \begin{itemize}
  \item an edge $\{w,\low\}$ of weight $\wfl(i)=2^{2i+1}\cdot i\cdot n$ for every $w \in W^i_{\lowl}$ and every
    $i\in [k]$,
  \item an edge $\{w,\high\}$ of weight $\wfh(i)=2^{2i}\cdot i\cdot n$ for every $w \in W^i_{\highl}$ and every
    $i\in [k]$,    
  \item an edge $\{\high,v^i_h\}$ of weight $h+\wfh(i)$
    for every $i \in [k]$ and $h \in [n]$,
  \item an edge $\{\low,v^i_h\}$ of weight $n-h+1+\wfl(i)$
    for every $i \in [k]$ and $h \in [n]$,
  \item an edge $\{x_{\highl}^i,w_e\}$ of weight
    $h+1+\wfh(i)$ for every $i \in [k]$ and non-edge
    $e=\{v_h^i,u\}\notin E(G)$ with $u \in V(G)\setminus V_i$.
  \item an edge $\{x_{\lowl}^i,w_e\}$ of weight
    $n-h+1+\wfl(i)$
    for every $i \in [k]$ and non-edge
    $e=\{v_h^i,u\}\notin E(G)$ with $u \in V(G)\setminus V_i$.
  \end{itemize}

    Finally, we set the cost $\Delta$ to
    $nk+\sum_{i=1}^k(\wfh(i)+\wfl(i))$.     
  This completes the construction of the instance $(X,d,k',\cost)$ of the
  \MSR{} problem. Note that the construction can be
  achieved in $\poly(|V|,k)$-time, i.e., polynomial-time.

    \begin{figure}[t]
    \centering
\begin{tikzpicture}[ 
    vertex/.style={circle, draw, fill=black, inner sep=2pt}, 
    every label/.style={font=\small},
    ellipsis/.style={}
] 

% V_i (left) %
\node[vertex,fill= Carmine, label=left:{$v^i_1$}] (vi1) at (0,1.3) {}; 
\node[vertex, fill= Carmine, label=left:{$v^i_2$}] (vi2) at (0,0.8) {}; 
\node[vertex, fill= Carmine, label=left:{$v^i_h$}] (vi3) at (0,0) {}; 
\node[vertex, fill= Carmine, label=left:{$v^i_n$}] (vin) at (0,-1.3) {}; 
% Replace dotted lines with ellipsis pattern
\path (vi2) -- (vi3) node[pos=0.5] {$\ldots$};
\path (vi3) -- (vin) node[pos=0.5] {$\ldots$};
\draw[
    thick,
    fill=Carmine,
    fill opacity=0.15
] (0,0) ellipse [x radius=1.2cm, y radius=1.8cm];
 
\node at (-1.6,0) {$V_i$}; 

% x^i centers %
\node[vertex,label=below:{$x^i_{\highl}$}] (xih) at (2,0.75) {}; 
\node[vertex,label=above:{$x^i_{\lowl}$}] (xil) at (2,-0.75) {};  

% W^i sets % 
\draw[fill=gray!2.5, rounded corners] (1.25,2.0) rectangle (2.75,3.0); 
\node at (2,3.3) {$W^i_{\highl}$}; 
\node[vertex] (wih1) at (1.5,2.5) {}; 
\node[vertex] (wih2) at (1.85,2.5) {}; 
\node[vertex] (wih3) at (2.35,2.5) {}; 
% Replace dotted line with ellipsis
\path (wih2) -- (wih3) node[pos=0.5] {$\ldots$};

\draw[fill=gray!2.5, rounded corners] (1.25,-3.0) rectangle (2.75,-2.0);
\node[vertex] (wil1) at (1.5,-2.5) {}; 
\node[vertex] (wil2) at (1.85,-2.5) {}; 
\node[vertex] (wil3) at (2.35,-2.5) {}; 
% Replace dotted line with ellipsis
\path (wil2) -- (wil3) node[pos=0.5] {$\ldots$};
\node at (2,-3.25) {$W^i_{\lowl}$};

% edges from x^i_high to W^i_high
\draw (xih) -- (wih1);
\draw (xih) -- (wih2);
\draw (xih) -- (wih3);

% edges from x^i_low to W^i_low
\draw (xil) -- (wil1);
\draw (xil) -- (wil2);
\draw (xil) -- (wil3);

 % E_{i,j} %
\draw[
    thick,
    rounded corners
] (3.4,-2.3) rectangle (4.6,2.3);

\node at (4,2.6) {$E_{i,j}$}; 
\node[vertex] (e1) at (4,2) {}; 
\node[vertex] (e2) at (4,1.5) {}; 
\node[vertex] (e3) at (4,1) {}; 
\node[vertex] (e4) at (4,0.5) {}; 
\node[vertex,label= $w_e$] (e6) at (4,-0.5) {}; 
\node[vertex] (e7) at (4,-1.5) {}; 
\node[vertex] (e8) at (4,-2) {};
% Replace dotted line with ellipsis
\path (e6) -- (e7) node[pos=0.5] {$\ldots$};
\path (e4) -- (e6) node[pos=0.3] {$\ldots$};
 %\draw[thick] (4,0) ellipse [x radius=1.2cm, y radius=1.8cm]; 

% x^j centers %
\node[vertex, label=below:{$x^j_{\highl}$}] (xjh) at (6,0.75) {}; 
\node[vertex, label=above:{$x^j_{\lowl}$}] (xjl) at (6,-0.75) {}; 

 % W^j sets %
\draw[fill=gray!2.5, rounded corners] (5.25,2.0) rectangle (6.75,3.0); 
\node at (6,3.3) {$W^j_{\highl}$}; 
\node[vertex] (wjh1) at (5.5,2.5) {}; 
\node[vertex] (wjh2) at (5.85,2.5) {}; 
\node[vertex] (wjh3) at (6.35,2.5) {}; 
% Replace dotted line with ellipsis
\path (wjh2) -- (wjh3) node[pos=0.5] {$\ldots$};
\draw[fill=gray!2.5, rounded corners] (5.25,-3.0) rectangle (6.75,-2.0); 
\node at (6,-3.2) {$W^j_{\lowl}$}; 
\node[vertex] (wjl1) at (5.5,-2.5) {}; 
\node[vertex] (wjl2) at (5.85,-2.5) {}; 
\node[vertex] (wjl3) at (6.35,-2.5) {}; 
% Replace dotted line with ellipsis
\path (wjl2) -- (wjl3) node[pos=0.5] {$\ldots$};

% edges from x^j_high to W^j_high
\draw (xjh) -- (wjh1);
\draw (xjh) -- (wjh2);
\draw (xjh) -- (wjh3);
% edges from x^j_low to W^j_low
\draw (xjl) -- (wjl1);
\draw (xjl) -- (wjl2);
\draw (xjl) -- (wjl3);

% V_j (right) %
\node[vertex, fill=BrownDrab, label=right:{$v^j_1$}] (vj1) at (8,1.3) {};
\node[vertex, fill=BrownDrab, label=right:{$v^j_2$}] (vj2) at (8,0.8) {};
\node[vertex, fill=BrownDrab, label=right:{$v^j_{h'}$}] (vj3) at (8,0) {};
\node[vertex, fill=BrownDrab, label=right:{$v^j_n$}] (vjn) at (8,-1.3) {};
% Replace dotted lines with ellipsis
\path (vj2) -- (vj3) node[pos=0.5] {$\ldots$};
\path (vj3) -- (vjn) node[pos=0.5] {$\ldots$};
\draw[
    thick,
    fill=BrownDrab,
    fill opacity=0.15
] (8,0) ellipse [x radius=1.2cm, y radius=1.8cm];

\node at (9.6,0) {$V_j$};

 % Correct edges %------------------------
%x^i to V_i% 
\draw[thick] (xih) -- (vi1); 
\draw[thick] (xih) -- (vi2); 
\draw[thick] (xih) -- (vi3); 
\draw[thick] (xih) -- (vin); 
\draw[thick] (xil) -- (vi1); 
\draw[thick] (xil) -- (vi2); 
\draw[thick] (xil) -- (vi3); 
\draw[thick] (xil) -- (vin); 

% x^i to E_{i,j} (partial adjacency)%
\draw[thick] (xih) -- (e1); 
\draw[thick] (xih) -- (e2); 
\draw[thick] (xih) -- (e3); 
\draw[thick] (xih) -- (e4); 
\draw[thick] (xih) -- (e6); 
\draw[thick] (xih) -- (e7); 
\draw[thick] (xih) -- (e8); 
\draw[thick] (xil) -- (e1); 
\draw[thick] (xil) -- (e2); 
\draw[thick] (xil) -- (e3); 
\draw[thick] (xil) -- (e4); 
\draw[thick] (xil) -- (e6); 
\draw[thick] (xil) -- (e7); 
\draw[thick] (xil) -- (e8);

% x^j to E_{i,j} (different neighbors) 
\draw[thick] (xjh) -- (e1); 
\draw[thick] (xjh) -- (e2); 
\draw[thick] (xjh) -- (e3); 
\draw[thick] (xjh) -- (e4); 
\draw[thick] (xjh) -- (e6); 
\draw[thick] (xjh) -- (e7); 
\draw[thick] (xjh) -- (e8); 
\draw[thick] (xjl) -- (e1); 
\draw[thick] (xjl) -- (e2); 
\draw[thick] (xjl) -- (e3); 
\draw[thick] (xjl) -- (e4); 
\draw[thick] (xjl) -- (e6); 
\draw[thick] (xjl) -- (e7); 
\draw[thick] (xjl) -- (e8); 

% x^j to V_j 
\draw[thick] (xjh) -- (vj1); 
\draw[thick] (xjh) -- (vj2); 
\draw[thick] (xjh) -- (vj3); 
\draw[thick] (xjh) -- (vjn); 
\draw[thick] (xjl) -- (vj1); 
\draw[thick] (xjl) -- (vj2); 
\draw[thick] (xjl) -- (vj3); 
\draw[thick] (xjl) -- (vjn); 

\end{tikzpicture}
\caption{ Gadget illustrating the interaction between color classes $V_i$ and $V_j$ in the reduction. 
For each $i \in [k]$, the two potential centers $x^i_{\highl}$ and $x^i_{\lowl}$ are connected to all vertices of $V_i$ and to two sets of $2k+1$ leaves $W^i_{\highl}$ and $W^i_{\lowl}$, which enforce the selection of both centers via large-radius penalties. 
Edges from $x^i_{\highl}$ (resp.\ $x^i_{\lowl}$) to $v_h^i \in V_i$ have weight $h+\omega_{\highl}(i)$ (resp.\ $n-h+1+\omega_{\lowl}(i)$), and edges to leaves in $W^i_{\highl}$ (resp.\ $W^i_{\lowl}$) have weight $\omega_{\highl}(i)$ (resp.\ $\omega_{\lowl}(i)$). 
For every non-edge $e=\{u,v\}\notin E(G)$ with $u=v_h^i \in V_i$ and $v= v^j_{h'} \in V_j$, a vertex $w_e$ (shown in the box $E_{i,j}$) is added and connected to $x^i_{\highl},x^i_{\lowl}$ with weights $h+1+\omega_{\highl}(i)$ and $n-h+1+\omega_{\lowl}(i)$ (and symmetrically to $x^j_{\highl},x^j_{\lowl}$ for $j$). 
Selecting incompatible vertices in $V_i$ and $V_j$ forces some $w_e$ to be covered at a prohibitively large radius, while a multicolored clique yields a feasible clustering of total cost at most $\Delta$.}
 \label{fig:gadget-eij}
\end{figure}
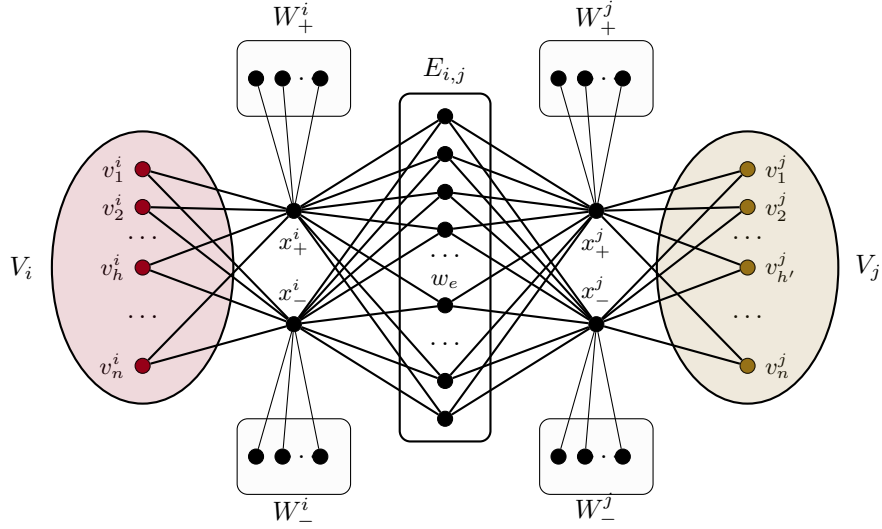
  
 \begin{intuitionbox}
The core strategy of the reduction is to treat the $2k$ cluster centers as rigid vertex \emph{selectors}, and their chosen radii as the index \emph{decisions}. 

To simulate selecting exactly one vertex $v_{h_i}^i$ from each color class $V_i$, we introduce a pair of anchor vertices for each class: a high anchor $x^i_{\highl}$ and a low anchor $x^i_{\lowl}$. To tightly lock the centers onto these anchors, we attach sets of $2k+1$ leaf vertices ($W^i_{\highl}, W^i_{\lowl}$) to them. Because the number of clusters is strictly limited to $2k$, at least one leaf vertex in each set cannot be chosen as a center. Covering these unchosen leaves within the budget forces every single anchor to be selected as a cluster center. Crucially, the total budget allocated to the radii of these two anchors forms an exact zero-sum trade-off: if the radius of $x^i_{\highl}$ is expanded to $h_i + \wfh(i)$ to cover a prefix of vertices $\{v_1^i, \dots, v_{h_i}^i\}$, the radius of $x^i_{\lowl}$ is compressed to exactly $(n - h_i) + \wfl(i)$, which is just enough to cover the complementary suffix $\{v_{h_i+1}^i, \dots, v_n^i\}$. This prefix-suffix split structurally locks in the choice of a single index $h_i \in [n]$ per color class.
\end{intuitionbox}
\begin{intuitionbox}
To validate that the vertices correspond to these chosen indices form a clique, we introduce a non-edge vertex $w_e$ for every non-edge $e \notin E(G)$ between color classes $V_i$ and $V_j$. The edge weights connecting these non-edge vertices to the anchors are specifically designed to mirror the prefix-suffix split. As a result, the anchors of color class $i$ will automatically cover $w_e$ \emph{unless} $e$ happens to be incident to the selected vertex $v_{h_i}^i$. Symmetrically, the anchors of class $j$ cover $w_e$ \emph{unless} it is incident to $v_{h_j}^j$. Therefore, if a non-edge exists between two selected vertices, neither class has enough radius to reach it, leaving $w_e$ completely exposed. The non-edge gadgets are fully covered if and only if the selected vertices share zero non-edges, signifying a valid multicolored clique.

Finally, to rigidly enforce this isolation and prevent metric "spillover" (where a cluster from color class $i$ maliciously helps cover vertices in a class $j$), the base weights $\wfh(i)$ and $\wfl(i)$ are scaled exponentially by a factor of $2^{2i} \cdot i \cdot n$. This exponential separation in scale \emph{overshadows} the local coordinates ($n$), meaning any cross-color coverage between different parts immediately violates the tight budget $\Delta$.
\end{intuitionbox}
  
  Moreover, $G'$ is bipartite and because
  $G'-\{x_{\highl}^1,x_{\lowl}^1,\dotsc,x^k_{\highl},x^k_{\lowl}\}$ is an independent set, it holds that the vertex cover number of $G'$ is at most $k'=2k$.
  It therefore only remains to show that $G$ has
  a multicolored $k$-clique if and only if $(X,d,k',\cost)$ has a clustering into at most
  $k'$
  clusters of cost at most $\cost$.

  We start by showing the following simple observation.
  \begin{observation}\label{obs:edge-dist}
    Let $\{u,v\}$ be an edge of $G'$ of weight $t>0$. Then, $d(u,v)\geq t$.
  \end{observation}
  \begin{proof}
%Let $\{u,v\}$ be an edge of $G'$ of weight $w$.
By construction, every edge of $G'$ is incident
to an anchor vertex of the form $x^i_{\highl}$ or $x^i_{\lowl}$ for some $i \in [k]$.
Hence, without loss of generality, we may assume that $u = x^i_{\highl}$ for some $i \in [k]$;
the case $u = x^i_{\lowl}$ is symmetric.

If $v \in W^i_{\highl}$, then $v$ has degree one in $G'$, and hence $\{u,v\}$ is the unique
$u$--$v$ path. Thus $d(u,v) = t$.

Otherwise, let $P$ be any $u$--$v$ path that does not use the edge $\{u,v\}$.
The first edge of $P$ leaves $u=x^i_{\highl}$, and by construction every edge incident to
$x^i_{\highl}$ has weight at least $\wfh(i)$. We distinguish the following two cases.

\smallskip
\noindent\emph{\textbf{Case 1: $v \in V_i$.}}
Then as $v$ is adjacent exactly to $u=x^i_{\highl}$ and $x^i_{\lowl}$, the last edge of $P$ must enter $v$ from $x^i_{\lowl}$, and
hence has weight at least $\wfl(i)$.
Therefore, $w(P) \ge \wfh(i)+ \wfl(i)$.

\smallskip
\noindent\emph{\textbf{Case 2: $v = w_e \in \overline{E}_{i,j}$ for some $j \in [k]$ and $j \neq i$.}}
Any $u$--$v$ path $P$ that avoids $\{u,v\}$ must enter $v$ from one of its remaining neighbors $x^i_{\lowl}, x^j_{\highl}$, or $x^j_{\lowl}$. 
If $P$ enters $v$ from $x^i_{\lowl}$, then the last edge has weight at least $\wfl(i)$, and
hence $w(P) \ge \wfh(i) + \wfl(i)$.
Otherwise, $P$ enters $v$ from $x^j_{\highl}$ or $x^j_{\lowl}$, in which case the last edge
has weight at least $\wfh(j)$ or $\wfl(j)$, respectively, and hence
$w(P) \ge \wfh(i) + \wfh(j)$. Therefore, $w(P) \ge \wfh(i) +\min\{\wfl(i), \wfh(j)\}$.

\smallskip
In both cases, by construction the weight of the direct $\{u,v\}$ edge satisfies
$t \le \wfh(i) + n + 1$, while for all $j \ge 1$ we have
$\wfh(j) > n + 1$ and $\wfl(i) > n+1$.
Hence $w(P) > t$, and therefore $d(u,v) \ge t$.
\end{proof}

  Towards showing the forward direction, let
 $S=\{v_{h_{1}}^1,v_{h_2}^2,\ldots,v_{h_k}^k\}$ be a multicolored
  clique of size $k$ in $G$ for some $h_i \in [n]$ for every $i \in [k]$. We obtain a
  corresponding clustering $\CCC$ for $(X,d,k',\cost)$ as follows. For each
  $1\leq i\leq k$, the clustering  $\CCC$ contains two clusters $C^i_{\lowl}$ and
  $C^i_{\highl}$ with centers $x^i_{\lowl}$ and
  $x^i_{\highl}$ and radii $r^i_{\lowl}=n-h_i+2^{2i+1}\cdot i\cdot n$ and
  $r^{i}_{\highl}=h_i+2^{2i}\cdot i\cdot n$, respectively.

  Clearly, the number of clusters $k'=2k$
  %\todo{\textcolor{white}{Pankaj: Have we defined $|\mathcal{C}|$}} 
  and the overall cost of the clustering $\mathcal{C}$ is 
  \[ \sum_{i\in [k]}\bigg((r^i_{\highl})+(r^i_{\lowl})\bigg)=\sum_{i\in [k]}\bigg((h_i+2^{2i}\cdot i\cdot n)+(n-h_i+2^{2i+1}\cdot i\cdot n)\bigg)= kn+\sum_{i\in [k]}3\cdot (2^{2i}\cdot i\cdot n).\] 
  It therefore merely remains to show
  that for every vertex $v \in X=V(G')$, there is a cluster in $\CCC$
  containing $v$. We first show that the vertices in $V_i$'s, $W^i_{\lowl}$, and $ W^i_{\highl}$ are covered.
  \begin{itemize}
  \item $W^i_{\lowl} \subseteq C^i_{\lowl}$ because the weight $2^{2i+1}\cdot i\cdot n$ of
    any edge $\{\low,w\}$ for $w \in W^i_{\lowl}$ is at most the radius
    $r^i_{\lowl}=n-h_i+2^{2i+1}\cdot i\cdot n$ of $C^i_{\lowl}$ (as $1\leq h_i\leq n$).
  \item $W^i_{\highl} \subseteq C^i_{\highl}$ because the weight $2^{2i}\cdot i\cdot n$ of
    any edge $\{\high,w\}$ for $w \in W^i_{\highl}$ is at most the radius
    $r^i_{\highl}=h_i+2^{2i}\cdot i\cdot n$ of $C^i_{\highl}$ (as  $1\leq h_i\leq n$).
  \item $\{v_1^i,\dotsc, v_{h_i}^i\} \subseteq C^i_{\highl}$ because
    the weight
    of any edge $\{x_{\highl}^i,v^i_{h}\}$ is $h+2^{2i}\cdot
    i\cdot n$ and that the radius of $C^i_{\highl}$ is $r^{i}_{\highl}=h_i+2^{2i}\cdot
    i\cdot n$.
  \item $\{v_{h_i+1}^i,\dotsc, v_{n}^i\} \subseteq C^i_{\lowl}$ because
    the weight
    of any edge $\{x_{\lowl}^i,v^i_{h}\}$ is $n-h+1+2^{2i+1}\cdot
    i\cdot n$  and that the radius of $C^i_{\lowl}$ is $r^{i}_{\lowl}=n-h_i+2^{2i+1}\cdot
    i\cdot n$.
  \end{itemize}
  It remains to show that the vertices in $\overline{E}_{i,j}$ are covered by
  $\CCC$. Towards showing this, we observe the following.

  \begin{observation}\label{obs:non-neighbor-exclusion}
    Let $i\in[k]$ and let $v^i_{h}\in V_i$. Let $C^i_{\highl}$ and
    $C^i_{\lowl}$ be the clusters with centers $x^i_{\highl}$ and
    $x^i_{\lowl}$ and radii $r^i_{\highl} = h + 2^{2i}\cdot i\cdot n$
    and $r^i_{\lowl} = n-h + 2^{2i+1}\cdot i\cdot n$, respectively.
    Then, for every $j \in [k]\setminus \{i\}$, it holds that $(C^i_{\highl}\cup C^i_{\lowl})\cap \overline{E}_{i,j}=\overline{E}_{i,j}\setminus
    \SB w_{\{v_h^i,u\}} \in \overline{E}_{i,j} \SM u \in V_j\SE$,
    i.e., $C^i_{\highl}\cup C^i_{\lowl}$ cover all vertices in
    $\overline{E}_{i,j}$ apart from the ones that correspond to non-edges
    incident to $v_h^i$. 
  \end{observation}
  \begin{proof}
    Consider a vertex $w_e \in \overline{E}_{i,j}$ with $e=\{v_\ell^i,u\}$ for
    some $\ell \in [n]$ and $u\in V_j$. Then, the edge $\{\high,w_e\}$ has weight
    $\ell+1+2^{2i}\cdot i\cdot n$ and because $r^i_{\highl}=h +
    2^{2i}\cdot i\cdot n$, it follows from \Cref{obs:edge-dist} that
    $w_e \in C^i_{\highl}$ if and only if $h\geq \ell+1$. Therefore,
    $C^i_{\highl}\cap \overline{E}_{i,j}=\SB w_{\{v_\ell^i,u\}}\in \overline{E}_{i,j}\SM
    \ell <h \land u \in V_j\SE$. Similarly, the edge $\{\low,w_e\}$ has weight
    $n-\ell+1+2^{2i+1}\cdot i\cdot n$ and because $r^i_{\lowl}=n-h +
    2^{2i+1}\cdot i\cdot n$, it follows from \Cref{obs:edge-dist} that
    $w_e \in C^i_{\lowl}$ if and only if $n-h\geq n-\ell+1$, i.e., $h\leq \ell-1$. Therefore,
    $C^i_{\lowl}\cap \overline{E}_{i,j}=\SB w_{\{v_\ell^i,u\}}\in \overline{E}_{i,j}\SM
    \ell > h \land u \in V_j\SE$. %This concludes the proof of the observation.
  \end{proof}
  Now consider the vertices in $\overline{E}_{i,j}$. By
  \Cref{obs:non-neighbor-exclusion}, $C^i_{\highl}\cup C^i_{\lowl}$
  cover all vertices in $\overline{E}_{i,j}$ apart from the vertices
  corresponding to non-edges incident to $v_{h_i}^i$. Similarly, 
  $C^j_{\highl}\cup C^j_{\lowl}$
  cover all vertices in $\overline{E}_{i,j}$ apart from the vertices
  corresponding to non-edges incident to $v_{h_j}^j$. Because $S$ is a multicolored
  $k$-clique in $G$, it holds that $v_{h_i}^i$ and $v_{h_j}^j$ are
  adjacent in $G$ and therefore there is no non-edge in $G$ that is
  both incident to $v_{h_i}^i$ and $v_{h_j}^j$ and hence
  $\overline{E}_{i,j} \subseteq C^i_{\highl}\cup C^i_{\lowl}\cup C^j_{\highl}\cup
  C^j_{\lowl}$.

  For the reverse direction, let $\CCC$ be a solution (clustering) for
  $(X,d,k',\Delta)$. We start by showing that $\CCC$ has $k'=2k$ clusters
  and those use the vertices $x_{\highl}^1,x_{\lowl}^1,\dotsc,x_{\highl}^k,x_{\lowl}^k$ as centers with certain radii.

\begin{claim}\label{clm:MainClm}
Let $\CCC$ be a solution for $(X,d,k',\Delta)$. 
Then $\CCC$ consists of exactly $k'$ clusters with centers 
$\{x_{\highl}^1, x_{\lowl}^1, \dotsc, x_{\highl}^k, x_{\lowl}^k\}$ and no other clusters. 
Moreover, for every $i \in [k]$ there is an $h_i \in [n]$ such that the cluster centered at $x_{\highl}^i$ has radius 
$r_{\highl}^i \ge h_i + \wfh(i)$, and the cluster centered at $x_{\lowl}^i$ has radius 
$r_{\lowl}^i \ge n - h_i + \wfl(i)$ (recall that $\wfh(i) = 2^{2i} \cdot i \cdot n$ and 
$\wfl(i) = 2^{2i+1} \cdot i \cdot n$).
\end{claim}

\begin{proof}
    Let $\ell \in [k]$ be the largest index such that there is no $h
    \in [n]$ such that $\CCC$ has clusters centered either at
    $x^\ell_{\highl}$ with radius
    $r^\ell_{\lowl}\geq h+\wfh(\ell)$ or at $x^\ell_{\lowl}$ with radius $r^{\ell}_{\lowl}\geq
    n-h+\wfl(\ell)$, respectively. Note that the statement of the claim
    now already applies for every $\ell'\in [k]$ with $\ell'>\ell$. We distinguish the
    following cases.

    \noindent\emph{\textbf{Case 1: $\CCC$ has no cluster with center
      $x^\ell_{\lowl}$.}}
    Because $|W^\ell_{\lowl}|>k'$, there is a vertex, say $w \in
    W^\ell_{\lowl}$, that is not the center of any cluster in
    $\CCC$. Let $C$ be the cluster in $\CCC$ that contains $w$. Then, as $w$ is not the center in $\CCC$ and has $x^\ell_{\lowl}$ as its unique neighbor,
    $C$ must also contain $x^\ell_{\lowl}$, but since $x^\ell_{\lowl}$
    is not the center of $C$ and every edge incident to
    $x^\ell_{\lowl}$ has cost at least $\wfl(\ell)$, we obtain that $C$
    has radius at least $2\wfl(\ell)$. 
    Because the statement of the claim holds for every $i>\ell$, we obtain that the clusters with centers $x^i_{\highl}$ and $x^i_{\lowl}$ exists for every $i$ where $\ell<i\leq k$ and have total cost at least
    $(n(k-\ell)+\sum_{i>\ell}
    \wfh(i)+\wfl(i))$. Moreover, even if $C$ is one of these clusters, it is easy to see that in order to cover $w$ the radius of such a cluster would have to be increased by at least $2\wfl(\ell)-n$ above the lower bound stated in the claim. Therefore, the 
    the cost of $\CCC$ is
    at least $(n(k-\ell-1)+\sum_{i>\ell}
    \wfh(i)+\wfl(i))+2\wfl(\ell)$. 
    Because
    $\cost=nk+\sum_{i=1}^k(\wfh(i)+\wfl(i))$,
    it suffices to show that $\cost=nk+\sum_{i=1}^k(\wfh(i)+\wfl(i))<(n(k-\ell-1)+\sum_{i>\ell}
    \wfh(i)+\wfl(i))+2\wfl(\ell)$, which is equivalent to showing that
    $n(\ell+1)+\sum_{i=1}^{\ell-1}(\wfh(i)+\wfl(i))+\wfh(\ell)<\wfl(\ell)$,
    which is shown as follows. We start by  expanding the left hand term

   % \begin{align*}
    %  & n(\ell+1)+\wfh(\ell)+\sum_{i=1}^{\ell-1}(\wfh(i)+\wfl(i))\\
     % = & n(\ell+1)+2^{2\ell}\ell n+\sum_{i=1}^{\ell-1}(2^{2i}i n+2^{2i+1}i n)\\
      %\leq & n(\ell+1)+2^{2\ell}\ell n+n(\ell-1)\sum_{i=2}^{2\ell-1}(2^{i})\\
      %\leq & n(\ell+1)+2^{2\ell}\ell
       %      n+n(\ell-1)(\sum_{i=0}^{2\ell-1}(2^{i})-3)\\
      %< & n(\ell+1)+2^{2\ell}\ell n+n(\ell-1)(2^{2\ell}-3)\\
     % = & n(\ell+1)+2^{2\ell+1}\ell n-2^{2\ell}n-3n(\ell-1)\\
      %= & n(\ell+1)+2^{2\ell+1}\ell n-3n\ell+3n-2^{2\ell}n\\
      %\leq & n(\ell+1)+2^{2\ell+1}\ell n-3n\ell-n\\
      %< & 2^{2\ell+1}(\ell+1) n\tag{\textcolor{cyan}{WHY WE HAVE ($\ell+1$)? looks like a typo!}}\\
      %= & \wfl(\ell)
    %\end{align*}

    \begin{align*}
      & n(\ell+1)+\wfh(\ell)+\sum_{i=1}^{\ell-1}(\wfh(i)+\wfl(i))\\
      = & n(\ell+1)+2^{2\ell}\cdot \ell \cdot n+\sum_{i=1}^{\ell-1}(2^{2i}\cdot i \cdot n+2^{2i+1}\cdot i \cdot n)\\
      = & n(\ell+1)+2^{2\ell}\cdot \ell\cdot n+\sum_{i=1}^{\ell-1}(3\cdot 2^{2i}\cdot i\cdot n)\\
      \leq & n(\ell+1)+2^{2\ell}\cdot\ell\cdot n+3n(\ell-1)\sum_{i=1}^{\ell-1}( 2^{i})\tag{by replacing $i$ with its maximum value}\\
= & n(\ell+1)+2^{2\ell}\cdot\ell\cdot n+3n(\ell-1)\big(\frac{2^{2\ell}-4}{3}\big)\tag{by putting the value of the term $\sum_{i=1}^{\ell-1}( 2^{i})$}\\
= & n(\ell+1)+2^{2\ell}\cdot\ell\cdot n+n(\ell-1)(2^{2\ell}-4)\\
= & n(\ell+1)+2^{2\ell}\cdot\ell\cdot n+2^{2\ell}\cdot \ell\cdot n -2^{2\ell}\cdot n-4\cdot \ell \cdot n+4n\\
= & n(\ell+1)+2^{2\ell+1}\cdot\ell\cdot n -2^{2\ell}\cdot n-4\cdot \ell \cdot n+4n\\
= & 2^{2\ell+1}\cdot\ell\cdot n -2^{2\ell}\cdot n-3\cdot \ell \cdot n+5n\\
= & \wfl(\ell) -2^{2\ell}\cdot n-3\cdot \ell \cdot n+5n\\
\leq  & \wfl(\ell) -2^{2\ell}\cdot n+2n\tag{as $ -3\cdot \ell\cdot n+5n \leq 2n$ for all $\ell\geq 1$}\\
< & \wfl(\ell)
\end{align*}
where the last inequality follows as $-2^{2\ell}\cdot n+2n<0$ for all $\ell\geq 1$.
    
    \noindent\emph{\textbf{Case 2: $\CCC$ has no cluster with center
      $x^\ell_{\highl}$, but $\CCC$ has a cluster with center $x^\ell_{\lowl}$.}}
    Because $|W^\ell_{\highl}|>k'$, there is a vertex, say $w \in
    W^\ell_{\highl}$, that is not the center of any cluster in
    $\CCC$. Let $C$ be the cluster in $\CCC$ that contains $w$. Then,
    $C$ must also contain $x^\ell_{\highl}$, but since $x^\ell_{\highl}$
    is not the center of $C$ and every edge incident to
    $x^\ell_{\highl}$ has cost at least $\wfh(i)$, we obtain that $C$
    has radius at least $2\wfh(i)$. Moreover, the radius of the
    cluster, say $C'\in \CCC$, with center $x^i_{\lowl}$ is at least $\wfl(i)$, since
    otherwise $C'$ only contains $x^i_{\lowl}$ and at least one of the
    clusters containing a vertex from $W^i_{\lowl}$ must already
    contain $x^i_{\lowl}$ (because $|W^i_{\lowl}|>k'$).
    Therefore, the cost of $\CCC$ is
    at least $(n(k-\ell)+\sum_{\ell'>\ell}
    \wfh(i)+\wfl(i))+2\wfh(i)+\wfl(i)=(n(k-\ell)+\sum_{\ell'>\ell}
    \wfh(i)+\wfl(i))+2\wfl(i)$ and therefore the same lower bound on
    the cost obtained in \textbf{Case 1}. Therefore, using the same
    calculation as in \textbf{Case 1}, we obtain that this cost
    lower bound is already larger than $\cost$.

    \noindent\emph{\textbf{Case 3: $\CCC$ has clusters with centers
      $x^\ell_{\highl}$ and $x^\ell_{\lowl}$.}}
    Note that because of \textbf{Case 1} and \textbf{Case 2}, we can
    from now on assume that $\CCC$ has a cluster with center $c$ for
    every $c \in \{x_{\highl}^1,x_{\lowl}^1, \dotsc,
    x_{\highl}^k,x_{\lowl}^k\}$ and because $|\{x_{\highl}^1,x_{\lowl}^1, \dotsc,
    x_{\highl}^k,x_{\lowl}^k\}|=k'$, $\CCC$ does not have any
    other clusters. Moreover, if $C \in \CCC$ is a cluster with center
    $x^i_{\highl}$ (or $x^i_{\lowl}$), then its radius $r^i_{\highl}$
    ($r^i_{\lowl}$) must be at least $\wfh(i)$ ($\wfl(i)$). This is
    because otherwise $C$ only contains $x^i_{\highl}$ ($x^i_{\lowl}$)
    but because $|W^i_{\highl}|>k'$ ($|W^i_{\lowl}|>k'$), there must be
    another cluster containing some $w \in W^i_{\highl}$ ($w \in W^i_{\lowl}$) and
    $x^i_{\highl}$ ($x^i_{\lowl}$) and therefore $C$ would not be
    needed. It follows that the cost of the clustering is already at
    least $n(k-\ell)+\sum_{i=1}^k(\wfh(i)+\wfl(i))$, leaving
    only $n\ell$ for additional costs.

    Because the conditions of the claim are not satisfied for $\ell$,
    it must now holds that
    $r^\ell_{\highl}+r^\ell_{\lowl}<\wfh(\ell)+\wfl(\ell)+n$. But then, there must
    be some $h \in \{0\}\cup [n-1]$ such that $r^\ell_{\highl}=\wfh(\ell)+h$ and
    $r^\ell_{\lowl}<n-h+\wfl(\ell)$, which implies that the vertex
    $v^i_{h+1}$ is not covered by $C_{\highl}^\ell$ or
    $C_{\lowl}^\ell$. But then, $v^i_{h+1}$ must be contained in some
    other cluster, say $C$, with center $c \in \{x_{\highl}^1,x_{\lowl}^1, \dotsc,
    x_{\highl}^k,x_{\lowl}^k\}\setminus \{x^\ell_{\highl},x^\ell_{\lowl}\}$.
    Assume that $c=x^i_{\lambda}$ for some $i \in
    [k]\setminus\{\ell\}$ and $\lambda \in \{\highl,\lowl\}$. Then, in
    order to cover $v^i_{h+1}$, the radius of $C$ must be at least by
    $2\wfh(\ell)$ larger than the lower bound that we computed for $C$
    above, i.e., the cost of the clustering is now at least
    $n(k-\ell)+\sum_{i=1}^k(\wfh(i)+\wfl(i))+2\wfh(\ell)$.
    Since $2\wfh(\ell)=2\cdot2^{2\ell}\cdot\ell \cdot n>n\ell$, this exceeds the
    allowed cost of $\cost=nk+\sum_{i=1}^\ell(\wfh(i)+\wfl(i))$ and
    concludes the proof of the claim.
  \end{proof}

  Because of \autoref{clm:MainClm}, we obtain that $\CCC$ contains
  exactly one cluster $C_{c}$ with center $c$ for every $c \in \{x_{\highl}^1,x_{\lowl}^1, \dotsc,
  x_{\highl}^k,x_{\lowl}^k\}$ and no other clusters. Moreover, for every $i \in [k]$,
  there is an $h_i \in [n]$ such that the cluster in $\CCC$ with
  center at $x_{\highl}^i$ or $x_{\lowl}^i$ has radius
  $r^i_{\highl}\geq h_{i}+\wfh(i)$ or $r^i_{\lowl}\geq
  n-h_{i}+\wfl(i)$, respectively. 
  Since the total cost is at most $\Delta$, these bounds are tight, and hence $r^i_{\highl}= h_{i}+\wfh(i)$ and $r^i_{\lowl}=
  n-h_{i}+\wfl(i)$ for all $i \in [k]$.
  We now claim that the vertex set
  $S=\{v^1_{h_1},\dotsc,v^k_{h_k}\}$ is a $k$-clique of $G$. In other
  words, it is sufficient to show that $\{v^i_{h_i},v^j_{h_j}\} \in
  E(G)$ for every $1 \leq i < j \leq k$. Because of
  \Cref{obs:non-neighbor-exclusion}, we obtain that for every $\ell
  \in \{i,j\}$, it holds that
  $C_{x^\ell_{\highl}}\cup C_{x^\ell_{\lowl}}$ cover exactly all vertices in
  $\overline{E}_{i,j}$ that do not correspond to non-edges incident to
  $v^\ell_{h_\ell}$. Therefore, $C_{x^i_{\highl}}\cup
  C_{x^i_{\lowl}}\cup C_{x^j_{\highl}}\cup
  C_{x^j_{\lowl}}$ cover all vertices in $\overline{E}_{i,j}$ if and only if $G$
  does not contain a non-edge between $v^i_{h_i}$ and $v^j_{h_j}$,
  i.e., if and only if $\{v^i_{h_i},v^j_{h_j}\} \in E(G)$. In the following, we observer that every vertex in $\overline{E}_{i,j}$ must be
  covered by one of the clusters $C_{x^i_{\highl}}$,
  $C_{x^i_{\lowl}}$, $C_{x^j_{\highl}}$, $C_{x^j_{\lowl}}$. Suppose this is not the case, i.e., there exists a vertex $w_e\in E_{i,j}$ that is not covered by one of the clusters $C_{x^i_{\highl}}$,
  $C_{x^i_{\lowl}}$, $C_{x^j_{\highl}}$, $C_{x^j_{\lowl}}$. As $\mathcal{C}$ is a valid clustering of $(X,d,k,\Delta)$ and hence $w_e$ must be covered by some other cluster centered at some $x^{t}_{\highl}$ (or at $x^{t}_{\lowl}$) where $t\in [k]\neq i,j$, which by construction must have radius more than $r^t_+=2^{2t}\cdot t\cdot n+h_{t'}$ (or more than $r^t_-=2^{2t+1}\cdot t\cdot n+n-h_{t'}$) for some $h_{t'}\in V_{t}$, a contradiction to \autoref{clm:MainClm}.
  This
  concludes the reverse direction of the proof.
\end{proof}
It is interesting to note that the result of \Cref{thm:vc+k} even
holds for the variant of \MSR, where the set of centers is given and
one merely needs to find the right radius for every center.

\subsection{Completing the Picture for the Weighted Case}%Employing Structural Parameters for Dense Graphs}
\label{ssec:wcomplete}
  
Here, we want to complete the picture for the parameterized complexity
of \MSR\ on metrics induced by  weighted graphs
w.r.t. structural parameters. As we have shown in~\Cref{thm:vc+k}, \MSR\ is
\W[1]-hard parameterized already by the vertex cover number, which is
arguably one of the most restrictive parameters in the context of
sparse graphs. Clearly, this already excludes fixed-parameter
algorithms for \MSR\ on metrics induced by weighted 
graphs parameterized by feedback vertex set number, treedepth,
treewidth, cliquewidth etc. However, it does not exclude
fixed-parameter algorithms for other structural parameters defined for
dense graphs such as neighborhood diversity, modular width, and
modular treedepth. The following result excludes such algorithms by
showing that \MSR\ is already \NP-hard on complete graphs and complete
bipartite graphs, i.e., graphs where all of these parameters are
constant.
Furthermore, it even excludes fixed-parameter
tractability when combining any of those structural parameters with
the parameter $k+\cost$.
\begin{theorem}\label{thm:MSR-complete}
  The \MSR\ problem is \textup{\NP\text{-hard}} and \textup{\W[1]\text{-hard}}
  parameterized by $k+\cost$ even on metrics
  induced by  weighted complete graphs and complete
  bipartite graphs.
\end{theorem}
\begin{proof}
  It follows from \Cref{thm:MT} that \MSR\ is
  \textup{\NP\text{-hard}} and \textup{\W[1]\text{-hard}}
  parameterized by $k+\cost$ even on metrics
  induced by  weighted bipartite graphs and we will
  show the theorem by providing a parameterized and polynomial-time reduction from this case to the two cases
  stated in the theorem. Let $(X,d,k,\cost)$ be an instance of the
  \MSR\ problem such that $X$ is the set of vertices of the
   weighted bipartite graph $G$ with partition
  $\{A,B\}$ and $d$ is the metric
  induced by $G$. To reduce to the case of a complete graph, let
  $(X,d',k,\cost)$ be the instance of \MSR\ such that $d'$ is the
  metric induced by the graph $G'$ obtained from $G$
  after adding an edge of cost $\cost+1$ between any two vertices in
  $G$ that are not adjacent.
  Similarly, to reduce to the case of a
  complete bipartite graph, let
  $(X,d'',k,\cost)$ be the instance of \MSR\ such that $d''$ is the
  metric induced by the graph $G''$ obtained from $G$
  after adding an edge of cost $\cost+1$ between any vertex in
  $A$ and any vertex in $B$.
  
  It is straightforward to show that these two
  constructions provide parameterized and polynomial-time
  reductions as required.
\end{proof}

While~\Cref{thm:MSR-complete} excludes fixed-parameter tractability
for any dense structural parameter plus $k+\cost$, we will now show
that \MSR\ on metrics induced by  weighted graphs is
fixed-parameter tractable parameterized by the treewidth of the graph
plus the cost of the clustering. To do so, we need the following
definitions.

Let $G$ be an  weighted graph and let $v \in V(G)$.
We denote by $N_r[v]$ the \emph{$r$-neighborhood} of $v$, i.e., the all vertices
that have distance at most $r$ to $v$. Moreover, we denote by
$\#(v)$ the number of distinct sets in $\SB N_r[v] \SM r \in
\mathbb{R}\SE$ and by $\#(G)$ the maximum of $\#(v)$ over all vertices $v
\in V(G)$. Using these notions, fixed-parameter tractability for
treewidth plus cost now follows immediately from the algorithm used to show
\cite[Theorem 1]{DBLP:journals/corr/abs-2310-02130} by observing that
the number of vertices $|V|$ in their statement of the run-time of the
algorithm can be easily replaced by the more concise parameter $\#(G)$.
\begin{proposition}[{\cite[Theorem 1]{DBLP:journals/corr/abs-2310-02130}}]\label{pro:XP-treewidth}
  The \MSR\ problem on metrics induced by  weighted
  graphs can be solved in time $\bigoh(\omega 2^{3\omega} (\#(G))^{3\omega+1}k^3)$, where $\omega$ is the treewidth of the
  graph $G$ underlying the metric.
\end{proposition}
Since $\#(G)$ is at most the cost $\cost$ , we obtain the following corollary from the
above proposition.
\begin{corollary}\label{cor:tw+Delta}
  The \MSR\ problem on metrics induced by  weighted
  graphs is fixed-parameter tractable parameterized by treewidth plus
  $\cost$.
\end{corollary}

\section{\MSR\ on metrics induced by (unweighted) graphs}

We now consider the case of (unweighted) graphs. As mentioned in the introduction, it is a
long standing open question whether \MSR\ on metrics induced by unweighted
graphs is solvable in polynomial-time~\cite{algorithmica/BehsazS15}. While we are not able to answer
this question, we provide the following main contributions. First, we
exhibit two natural variants of \MSR\ and provide strong hardness
results for these variants even on metrics induced by unweighted
 graphs. In particular, we introduce the
\exact\ version (where one is required to find a solution with exactly
$k$ non-zero clusters) in \Cref{ssec:exactnonzero} and show that it is
\NP-hard and \W[1]-hard parameterized by $k+\cost$. Moreover, we introduce the \fixed\
version (where one is additionally required to choose the clusters centers
from a given set of allowed centers) in \Cref{ssec:fixed} and show
that it is \W[1]-hard parameterized by either $k+\cost$ or $k+|A|+\fvs$,
where $A$ is the set of allowed centers and $\fvs$ is the feedback vertex
number of the underlying undirected graph. It is important to note that we are not able to show \NP-hardness for \fixed\ and we believe that our analysis for \fixed\ can be seen as a way around having to first settle the classical complexity of \MSR\ on graphs. That is, in order to obtain a parameterized complexity classification for \MSR\ it might be easier to show \W[1]-hardness as we do for \fixed. In addition to this, we also
show that \MSR\ as well as \exact\ and \fixed\ are fixed-parameter
tractable parameterized by the treedepth of the underlying graph.
We start by providing the algorithmic result and then introduce the
two variants and show their hardness in the following two subsections.

Let $G$ be an
undirected graph and let $\ell(G)$ denote the length of a longest path
in $G$. Then, $\#(G)\leq \ell(G)$.
It is well-known that the longest paths of an
undirected graph is asymptotically the same as the well-known
parameter treedepth and that the treewidth of a graph is at most equal
to its treedepth, i.e.:
\begin{proposition}[\cite{NesetrilMendez12}]
  Let $G$ be an undirected graph. Then, $\tw(G)\leq \td(G)$ and
  $\log \ell(G) \leq \td(G)\leq \ell(G)$.
\end{proposition}
Therefore, we obtain that $\#(G)\leq 2^\td(G)$ and $\tw(G)\leq
\td(G)$, which together with~\Cref{pro:XP-treewidth} implies (it is
straightforward to verify that the DP algorithm used to show
\Cref{pro:XP-treewidth} also works for \exact and \fixed):
\begin{corollary}\label{thm: treedepth}
  The \MSR\ problem (and the \exact\ and \fixed\ problems) on metrics
  induced by unweighted graphs is
  fixed-parameter tractable parameterized by the treedepth of the
  graph.
\end{corollary}

\subsection{\exact}\label{ssec:exactnonzero}

In this subsection, we study the version of the \MSR\ problem, which we call \exact\ that
asks for exactly $k$ clusters each containing at least two points
(this corresponds to requiring non-zero clusters).
Please refer to \autoref{sec:intro} for a formal definition of \exact.

We show that even on
metrics induced by an undirected bipartite graphs, \exact\ is
\NP-hard and \W[2]-hard parameterized by $k+\cost$.

\begin{restatable}{theorem}{thmExactMSR}\label{thm:EXACT-MSR}
  \exact\ on metrics induced by unweighted bipartite graphs is
    \textup{\NP\text{-hard}} and \textup{\W[2]\text{-hard}} parameterized by $k+\cost$. Moreover,
  assuming the Exponential Time Hypothesis, there is no algorithm for
  \exact\ running in time $f(k+\cost)\cdot n^{o(k+\cost)}$ for any
  computable function $f$.
\end{restatable}
\begin{proof}
  We provide a parameterized reduction from the \textsc{Dominating
    Set} problem on bipartite graphs that do not contain isolated vertices,
  which is well-known to be \NP-hard and \W[2]-hard
  when parameterized by the solution size~\cite{DowneyF92}.
  That is, given an instance $(G,k)$ of the \textsc{Dominating Set}
  problem (which asks whether the bipartite graph $G$ without isolated vertices has a dominating
  set of size at most $k$), we will construct the equivalent instance $(X,d,k,k)$
  of \exact{} in polynomial-time as follows. We set $X=V(G)$ and $d$
  is the metric induced by $G$. This completes the reduction, which
  can clearly be performed in polynomial time and it only remains to
  show that $(G,k)$ is a \yes-instance of \textsc{Dominating Set} if
  and only if $(X,d,k,k)$ has a solution.

  Towards showing the forward direction, let $D=\{d_1,\dotsc,d_\ell\}$
  be a dominating set of $G$ of size at most $\ell\leq k$. We claim
  that $\CCC=\{(d_1,1),\dotsc, (d_\ell,1),(v_1,1),\dotsc,(v_{k-\ell},1)\}$,
  where the $v_j$'s are arbitrary vertices in $G$ (that are pairwise
  distinct and distinct from all $d_i$') is a solution for
  $(X,d,k,k)$. Clearly, $|\CCC|=k$, the cost of $\CCC$ is at most $k$
  and since all center have at least one neighbor in $G$, it holds that
  at least two vertices are within distance at most $1$ from every center.
  Moreover, because $D$
  is a dominating set in $G$, it holds that every vertex in $G$ is
  within the radius of at least one of the clusters in $\CCC$.

  Towards showing the reverse direction, let $\CCC=\SB (x_i,r_i) \SM i
  \in [k]\SE$ be a solution for $(X,d,k,k)$. Because at least two
  vertices are within distance at most $r_i$ from $x_i$, it holds that
  $r_i\geq 1$. Moreover, since $\sum_{i \in[k]}r_i\leq \cost=k$, it
  holds that $r_i=1$. Therefore, every vertex of $G$ is a neighbor of
  some $x_i$ in $G$, which shows that the set $\{x_1,\dotsc,x_k\}$ is
  a dominating set of $G$.

  Finally, we note that the existence of a $f(k+\Delta)\cdot
  n^{o(k+\Delta)}$ time algorithm for any computable function $f$ for
  \exact\ would imply a
  $f(k) \cdot n^{o(k)}$ time algorithm for \textsc{Dominating Set}
  problem, a contradiction to \ETH, via~\cite[Corollary 14.23]{CyganFKLMPPS15}.
\end{proof}

\subsection{\fixed}\label{ssec:fixed}

In this subsection, we study the version of the \MSR\  problem, which we
coin \fixed,  where we
are additionally given a set of allowed centers.

Please refer to \autoref{sec:intro} for a formal definition of \fixed.
We start by showing that even on
metrics induced by an unweighted graph, \fixed\
parameterized by $k+\cost$ is \W[1]-hard.

\begin{restatable}{theorem}{thmFixedMSRcost}\label{thm:FixedMSR}
  \fixed\ on metrics induced by unweighted graphs is
  \textup{\W[1]\text{-hard}} parameterized by $k+\cost$. Moreover, there is no algorithm that solve \fixed\ in time $f(k+\Delta)\cdot n^{o(k+\Delta)}$ time for any computable function $f$, under \ETH.
\end{restatable}

%\thmFixedMSRcost*

\begin{proof}
  Consider the
  reduction from the \mcc{} problem to the \MSR{} problem provided in the
  proof of \Cref{thm:MT}. In particular, consider the
  instance $(X,d,k,\Delta)$ of \MSR\ together with the undirected
  edge-weighted graph $G'$ that induces the metric $(X,d)$ obtained
  from the instance $(G,k,(V_1,\dotsc,V_k))$ of \mcc\ in the proof of
  \Cref{thm:MT}. The idea is
  to transform $(X,d,k,\Delta)$ into an instance $(X',d',A,k,\Delta)$
  of \fixed{} such that the metric $(X',d')$ is induced by the
  undirected  graph $G''$, while preserving equivalence as
  follows. Let $G''$ be the undirected  graph obtained from
  $G'$ after subdividing every edge $e$ of $G'$ with weight $w$, $w-1$ many
  times. Moreover, let $X'=V(G'')$, let $d'$ be the metric induced by
  $G''$ and let $A=V(G)$ (recall that $V(G) \subseteq V(G')$).
  Because we
  are required to perform an exponential number of edge subdivisions,
  this is no longer a polynomial-time reduction (and
  therefore does not provide \NP-hardness) but it is a
  parameterized reduction for the parameters $k+\cost$. It is now
  straightforward to verify that $(G,k,(V_1,\dotsc,V_k))$ has a
  solution if and only if so does $(X',d',A,k,\Delta)$.

  Finally, we note that the existence of a $f(k+\Delta)\cdot
  n^{o(k+\log_2\Delta)}$ time algorithm for any computable function $f$ for
  \exact\ would imply a
  $f(k) \cdot n^{o(k)}$ time algorithm for \mcc\
  problem, a contradiction to \ETH, via~\cite[Corollary 14.23]{CyganFKLMPPS15}.
\end{proof}

The next theorem shows that \fixed\
parameterized by $k+|A|+\fvs$ is \W[1]-hard, where $A$ is the set of allowed centers and $\fvs$ is the feedback vertex set number.

\begin{restatable}{theorem}{thmFixedMSRfvs}\label{thm:FixedMSRfvs}
  \fixed\ on metrics induced by unweighted  graphs is
  \textup{\W[1]\text{-hard}} parameterized by $k+|A|+\fvs$. Moreover, there is no algorithm that solves \fixed\ in time $f(k+|A|+\fvs)\cdot n^{o(k+|A|+\fvs)}$ time for any computable function $f$, under \ETH.
\end{restatable}

%\thmFixedMSRfvs*

\begin{proof}
  Consider the
  reduction from the \mcc{} problem to the \MSR{} problem provided in the
  proof of \Cref{thm:vc+k}. In particular, consider the
  instance $(X,d,k',\Delta)$ of \MSR with $k'=2k$ together with the undirected
  edge-weighted graph $G'$ that induces the metric $(X,d)$ obtained
  from the instance $(G,k,(V_1,\dotsc,V_k))$ of \mcc\ in the proof of
  \Cref{thm:vc+k}. The idea is
  to transform $(X,d,k',\Delta)$ into an instance $(X',d',A,k',\Delta)$
  of \fixed{} such that the metric $(X',d')$ is induced by the
  undirected  graph $G''$, while preserving equivalence as
  follows. Let $G''$ be the undirected  graph obtained from
  $G'$ after subdividing every edge $e$ of $G'$ with weight $w$, $w-1$ many
  times. Moreover, let $X'=V(G'')$, let $d'$ be the metric induced by
  $G''$ and let $A=\SB\low, \high \SM i \in [k]\SE$.
  Note first that $A$ is a feedback vertex set for $G''$ and therefore
  the feedback vertex number of $G''$ is at most $k'=2k$.
  Moreover, because we
  are required to perform an exponential number of edge subdivisions,
  this is no longer a polynomial-time reduction (and
  therefore does not provide \NP-hardness) but it is a
  parameterized reduction for the parameters $k+|A|+\fvs$, where $\fvs$ is
  the feedback vertex set number of $G''$.
  It is now
  straightforward to verify that $(G,k,(V_1,\dotsc,V_k))$ has a
  solution if and only if so does $(X',d',A,k',\Delta)$.

 Finally, since $k + |A| + \fvs = 5k$, the existence of an algorithm for \exact\ running in time $f(k + |A| + \fvs) \cdot n^{o(k + |A| + \fvs)}$, for any computable function $f$, would imply an $f(k) \cdot n^{o(k)}$-time algorithm for the \mcc\ problem. This would contradict \ETH\ via~\cite[Corollary 14.23]{CyganFKLMPPS15}.
\end{proof}

\section{Conclusion}

Our main contribution is a comprehensive classification of the parameterized
complexity of \MSR\ for metrics induced by undirected weighted graphs
for any combination of the parameters $k$, $\cost$, any structural
parameter at least as restrictive as vertex cover, as well as the most common
structural parameters for dense graphs such as neighborhood diversity,
modular width, and cliquewidth. We also made some contribution for the unweighted
case by showing (parameterized) complexity lower bounds for two
natural and closely related variants. While the weighted case is now
fairly well understood, we think that it would be natural and interesting to
look at what happens to the parameterized complexity on planar graphs
(or generalizations of planar graphs), where neither of our two main
hardness results apply.

Finally, a major task for future work is 
to obtain a similarly detailed (parameterized) complexity landscape
for the unweighted case. Here, the main obstacle is the open
complexity status of \MSR\ on undirected unweighted graphs, i.e., only the
version with non-zero clusters is known to be solvable in
polynomial-time and for the general case there is a randomized quasi
polynomial-time algorithm, which makes it unlikely to be
\NP-hard. While it may be difficult to settle this question, our
\W[1]-hardness result for \fixed\ may offer a way out, since it suffices to show \W[1]-hardness instead of \NP-hardness to obtain a parameterized complexity classification. Note
that we leave it open whether \fixed\ is \NP-hard on undirected (unweighted) graph metrics, but are able to
show that the problem is \W[1]-hard.

\section*{Acknowledgements}

We want to thank the reviewers of SWAT 2026 for their helpful suggestions.
One of the authors thanks Tanmay Inamdar for helpful discussions on the
problem and for pointing out relevant references.

\bibliography{ref}
\appendix
\end{document}